%% file: paper.tex
\documentclass[11pt]{article}
\pdfoutput=1 
\usepackage{amssymb} 
\usepackage{fullpage}
\usepackage{graphicx}
\usepackage{color}
\usepackage{times}

\newcommand{\Rects}{R}
\newcommand{\ceil}[1]{\lceil{#1}\rceil}
\newcommand{\Skyline}{\mathrm{Skyline}} 
\newcommand{\Pred}{\mathrm{Pred}}
\newcommand{\Succ}{\mathrm{Succ}}

\newcommand{\rev}{\mathrm{rev}}
\newcommand{\bottom}{\mathrm{bottom}}
\newcommand{\ttop}{\mathrm{top}}
\newcommand{\Rightmost}{\mathrm{Rightmost}}
\newcommand{\select}{\mathrm{select}}

\newtheorem{corollary}{Corollary}
\newtheorem{theorem}{Theorem}
\newtheorem{lemma}{Lemma}
\newenvironment{proof}{\noindent\textit{Proof.} }{\mbox{}\hfill$\Box$}

\title{Optimal Planar Orthogonal Skyline Counting Queries%
\thanks{Full version of paper appearing in the proceedings of the 
  14th Scandinavian Symposium and Workshops on Algorithm Theory, 2014.}}

\author{Gerth St{\o}lting Brodal\thanks{MADALGO,
	Center for Massive Data Algorithmics, 
	a Center of the Danish National Research Foundation (grant DNRF84).
	Department of Computer Science, Aarhus University.
	$\{$\texttt{gerth,larsen}$\}$\texttt{@cs.au.dk}.}
 \and 
 	Kasper Green Larsen\footnotemark[1]}

\date{Arpil 24, 2014}

\begin{document}

\maketitle


\begin{abstract} 
  The skyline of a set of points in the plane is the subset of maximal
  points, where a point $(x,y)$ is maximal if no other point $(x',y')$
  satisfies $x'\geq x$ and $y'\geq y$. We consider the problem of
  preprocessing a set $P$ of $n$ points into a space efficient static
  data structure supporting orthogonal skyline counting queries,
  i.e.\ given a query rectangle $R$ to report the size of the skyline
  of $P\cap R$. We present a data structure for storing $n$ points
  with integer coordinates having query time $O(\lg n/\lg\lg n)$ and
  space usage $O(n)$ words.  The model of computation is a unit cost
  RAM with logarithmic word size.  We prove that these bounds are the
  best possible by presenting a matching lower bound in the cell probe
  model with logarithmic word size: Space usage $n\lg^{O(1)} n$
  implies worst case query time $\Omega(\lg n/\lg\lg n)$.
\end{abstract}

\section{Introduction}

In this paper we consider orthogonal range skyline queries for a set
of points in the plane. A point $(x,y)\in \mathbb{R}^2$
\emph{dominates} a point $(x',y')$ if and only if $x'\leq x$ and
$y'\leq y$. For a set of points~$P$, a point $p\in P$ is
\emph{maximal} if no other point in $P$ dominates $p$, and the
\emph{skyline} of~$P$, $\Skyline(P)$, is the subset of maximal points
in $P$.

We consider the problem of preprocessing a set $P$ of $n$ points in
the plane with integer coordinates into a data structure to support
\emph{orthogonal range skyline counting} queries: Given an
axis-aligned query rectangle $R=[x_1,x_2]\times[y_1,y_2]$ to report
the size of the skyline of the subset of the points from $P$ contained
in $R$, i.e.\ report $|\Skyline(P\cap R)|$. The main results of this
paper are matching upper and lower bounds for data structures
supporting such queries, thus completely settling the problem. Our
model of computation is the standard unit cost RAM with logarithmic
word size.

\subsection{Previous Work}

Orthogonal range searching is one of the most fundamental and
well-studied topics in computational geometry, see e.g.~\cite{socg11}
for an extensive list of previous results. For orthogonal range
queries in the plane, with integer coordinates in $[n] \times [n] =
\{0,\dots,n-1\} \times \{0,\dots,n-1\}$, the main results are the
following: For the orthogonal range \emph{counting} problem,
i.e.\ queries report the total number of input points inside a query
rectangle, optimal $O(\lg n/\lg\lg n)$ query time using $O(n)$ space
was achieved in~\cite{Jaja04}. Optimality was shown in \cite{stoc07},
where it was proved that space $n\lg^{O(1)} n$ implies query time
$\Omega(\lg n/\lg\lg n)$ for range counting queries.

For range \emph{reporting} queries it is known that space $n\lg^{O(1)}
n$ implies query time $\Omega(\lg\lg n+k)$, where $k$ is the number of
points reported within the query range \cite{stoc06}.  The best upper
bounds known for range reporting are: Optimal space $O(n)$ and query
time $O((k+1)\lg^\varepsilon n)$ \cite{socg11}, and optimal query time
$O(\lg\lg n+k)$ with space $O(n\lg^\varepsilon n)$~\cite{focs00}. In
both cases $\varepsilon>0$ is an arbitrarily small constant.

\paragraph{Orthogonal Range Skyline Queries.}

Orthogonal range skyline counting queries were first consider
in~\cite{walcom12}, where a data structure was presented with space
usage $O(n \lg^2n / \lg \lg n)$ and query time $O(\lg^{3/2}n/\lg\lg
n)$. This was subsequently improved to $O(n \lg n)$ space and $O(\lg
n)$ query time~\cite{iwoca12}. Finally, a data structure achieving an
even faster query time of $O(\lg n /\lg \lg n)$ was presented, however
the space usage of that solution was a prohibitive $O(n \lg^3 n/\lg
\lg n)$~\cite{walcom13}. Thus to date, no linear space solution exists
with a non-trivial query time. Also, from a lower bound perspective,
it is not known whether the problem is easier or harder than the
standard range counting problem.

For orthogonal skyline reporting queries, the best bound is $O(n\lg
n/\lg\lg n)$ space with query time $O(\lg n/\lg\lg
n+k)$~\cite{walcom12}, where $k$ is the size of the reported
skyline. Note that an $\Omega(\lg \lg n)$ search term is needed for
skyline range reporting since the $\Omega(\lg \lg n)$ lower bound for
standard range reporting was proved even for the case of determining
whether the query rectangle is empty~\cite{stoc06}.

\begin{table}[t]
\begin{minipage}[t]{7,3cm}
  \centering
  \caption{Previous and new results for skyline counting queries.}
  \label{tab:previous-counting}
  \vspace{2ex}
  \begin{tabular}{ccl}
    Space (words) & Query time & Reference \\
    \hline
    $n\frac{\lg^{2} n}{\lg\lg n}$ & $\frac{\lg^{3/2}n} {\lg\lg n}$ & \cite{walcom12} \\
    $n\lg n$ & $\lg n$ & \cite{iwoca12} \\
    $n \frac{\lg^3 n}{\lg\lg n}$ & $\frac{\lg n}{\lg\lg n}$ & \cite{walcom13} \\ 
    $n$ & $\frac{\lg n}{\lg\lg n}$ & \textbf{New} \\ 
    \hline
  \end{tabular}
\end{minipage}
\hfill
\begin{minipage}[t]{8,5cm}
  \centering
  \noindent \caption{Previous and new results for skyline reporting queries.}
  \label{tab:previous-reporting}
  \vspace{2ex}
  \begin{tabular}{ccl}
    Space (words) & Query time & Reference \\
    \hline
    $n\lg n$ & $\lg^2 n+k$ & \cite{icalp11} (dynamic) \\
    $n\lg n$ & $\lg n+k$ & \cite{cccg11,iwoca12} \\
    $n\frac{\lg n}{\lg\lg n}$ & $\frac{\lg n}{\lg\lg n}+k$ & \cite{walcom12}\\
    $n\lg^{\varepsilon} n$ & $(k+1)\lg\lg n$ &\cite{swat12} \\
    $n\lg^{\varepsilon} n$ & $\frac{\lg n}{\lg\lg n}+k$ & \textbf{New} \\
    $n\lg\lg n$ & $(k+1)(\lg\lg n)^2$ & \cite{swat12} \\
    $n\lg\lg n$ & $\frac{\lg n}{\lg\lg n}+k\lg\lg n$ & \textbf{New} \\
    $n$ & $(k+1)\lg^{\varepsilon} n$ & \cite{swat12} \\
    \hline
  \end{tabular}
\end{minipage}
\end{table}

In \cite{swat12} solutions for the sorted range reporting problem were
presented, i.e.\ the problem of reporting the $k$ leftmost points
within a query rectangle in sorted order of increasing
$x$-coordinate. With space $O(n)$, $O(n\lg\lg n)$ and
$O(n\lg^{\varepsilon} n)$, respectively, query times
$O((k+1)\lg^{\varepsilon} n)$, $O((k+1)(\lg\lg n)^2)$, and $O(k+\lg\lg
n)$ were achieved, respectively. The structures of \cite{swat12}
support finding the rightmost (skyline) point in a query range
($k=1$). By recursing on the rectangle above the reported point one
immediately get the bounds for skyline reporting listed in
Table~\ref{tab:previous-reporting}, where only the linear space
solution achieves query times matching those of general orthogonal
range reporting.

Previous results for skyline queries are summarized in
Tables~\ref{tab:previous-counting} and~\ref{tab:previous-reporting}.

\subsection{Our Results}

In Section~\ref{sec:upper-bound} we present a linear space data
structure supporting orthogonal range skyline counting queries in
$O(\lg n/\lg\lg n)$ time, thus for the first time achieving linear
space and improving over all previous tradeoffs. In
Section~\ref{sec:lower-bound} we show that this is the best possible
by proving a matching lower bound. More specifically, we prove a lower
bound stating that the query time $t$ must satisfy $t = \Omega(\lg
n/\lg(Sw/n))$. Here $S\ge n$ is the space usage in number of words and
$w=\Omega(\lg n)$ is the word size in bits. For $w=\lg^{O(1)} n)$ and
$S=n\lg^{O(1)}n$, this bound becomes $t=\Omega(\lg n/\lg \lg n)$. The
lower bound is proved in the cell probe model of
Yao~\cite{yao:cellprobe}, which is more powerful than the unit cost
RAM and hence the lower bound also applies to RAM data structures.

As a side result, we in Section~\ref{sec:reporting} show how to modify
our counting data structure to support reporting queries. Our
reporting data structure has query time $O(\lg n/\lg \lg n+k)$ and
space usage $O(n \lg^{\varepsilon} n)$. The best previous reporting
structure with a linear term in $k$ has $O(\lg n/\lg \lg n+k)$ query
time but $O(n \lg n/\lg \lg n)$ space~\cite{walcom12}. The reporting
structure can also be modified to achieve $O(\lg n/\lg \lg n + k \lg
\lg n)$ query time and $O(n \lg \lg n)$ space. See
Table~\ref{tab:previous-reporting} for a comparison to previous
results.

Our lower bound follows from a reduction of reachability in butterfly
graphs to two-sided skyline counting queries, extending reductions by
P\v{a}tra\c{s}cu~\cite{patrascu11structures} for two-dimensional
rectangle stabbing and range counting queries.  Our upper bounds are
achieved by constructing a balanced search tree of
degree~$\Theta(\lg^{\varepsilon} n)$ over the points sorted by
$x$-coordinate. At each internal node we store several space efficient
rank-select data structures storing the points in the subtrees sorted
by rank-reduced $y$-coordinates. Using a constant number of global
tables, queries only need to spend $O(1)$ time at each level of the
tree.

\subsection{Preliminaries}

\paragraph{Coordinates.}
If the coordinates of the input and query points are not restricted to
$[n] \times [n]$, but can be arbitrary integers that fit into a
machine word, then we can map the coordinates to the range $[n]$ by
using the RAM dictionary from \cite{stoc99}, which can support
predecessor queries on the lexicographical orderings of the points in
time $O(\sqrt{\lg n/\lg\lg n})$ using $O(n)$ space. This is less than
the $O(\lg n/\lg \lg n)$ query time we are aiming for.

\paragraph{Succinct Data Structures.}
In our solutions, we make extensive use of the following results from
succinct data structures.

\begin{lemma}[\cite{soda02}]
\label{lem:rankselect}
  A vector $X[1..s]$ of $s$ zero-one values, with $t$ values equal to
  one, can be stored in a data structure of size $O(t(1+\lg s/t))$
  bits supporting $\mathrm{rank}$ and $\mathrm{select}$ queries in
  $O(1)$ time. A $\mathrm{rank}(i)$ query returns the number of ones
  in $X[1..i]$, provided $X[i]=1$, whereas a $\mathrm{select}(i)$
  query returns the position of the $i$'th one in~$X$.
\end{lemma}

\begin{lemma}[\cite{talg07}]
\label{lem:prefixsum}
  Let $X[1..s]$ be a vector of $s$ non-negative integers with total
  sum $t$.  There exists a data structure of size $O(s\lg (2+t/s))$
  bits, supporting the lookup of $X[i]$ and the prefix sum
  $\sum_{j=1}^i X[j]$ in $O(1)$ time, for $i=1,\ldots,s$.
\end{lemma}

\begin{lemma}[\cite{jda07sadakane,latin10fischer}]
\label{lem:rmq}
  Let $X[1..s]$ be a vector of integers. There exists a data structure
  of size $O(s)$ bits supporting range-maximum-queries in $O(1)$ time,
  i.e.\ given $i$ and $j$, $1\leq i\leq j\leq s$, reports the index
  $k$, $i\leq k\leq j$, such that $X[k]=\max(X[i..j])$. Queries only
  access this data structure, i.e.\ the vector $X$ is not stored.
\end{lemma}

\section{Lower Bound}
\label{sec:lower-bound}

That an orthogonal range skyline counting data structure requires
space $\Omega(n\lg n)$ bits, follows immediately since each of the
$n!$ different input point sets of size~$n$, where points have
distinct $x$- and $y$-coordinates from~$[n]$, can be reconstructed
using query rectangles considering each possible point in $[n]^2$
independently, i.e.\ the space usage is at least $\ceil{\lg_2
  (n!)}=\Omega(n\lg n)$ bits.

In the remainder of this section, we prove that any data structure
using $S\ge n$ words of space must have query time $t = \Omega(\lg
n/\lg(Sw/n))$, where $w = \Omega(\lg n)$ denotes the word size in
bits. In particular for $w=\lg^{O(1)} n$, this implies that any data
structure using $n \lg^{O(1)}n$ space must have query time
$t=\Omega(\lg n/\lg \lg n)$, showing that our data structure from
Section~\ref{sec:upper-bound} is optimal. Our lower bound holds even
for data structures only supporting skyline counting queries inside
$2$-sided rectangles, i.e.\ query rectangles of the form
$(-\infty,x]\times(-\infty,y]$.  The lower bound is proved in the cell
    probe model of Yao~\cite{yao:cellprobe} with word size $w =
    \Omega(\lg n)$. Since we derive our lower bound by reduction, we
    will not spend time on introducing the cell probe model, but
    merely note that lower bounds proved in this model applies to data
    structures developed in the unit cost RAM model. See
    e.g.~\cite{stoc07} for a brief description of the cell probe
    model.

\paragraph{Reachability in the Butterfly Graph.}

We prove our lower bound by reduction from the problem known as
\emph{reachability oracles in the butterfly
  graph}~\cite{patrascu11structures}. A butterfly graph of degree $B$
and depth $d$ is a directed graph with $d+1$ layers, each having $B^d$
nodes ordered from left to right (see Figure~\ref{fig:butterfly}). The
nodes at level~$0$ are the \emph{sources} and the nodes at level $d$
are the \emph{sinks}. Each node, except the sinks, has out-degree $B$,
and each node, except the sources, has in-degree~$B$.

If we number the nodes at each level with $0,\dots,B^d-1$ from left to
right and interpret each index $i \in [B^d]$ as a vector
$v(i)=v(i)[d-1]\cdots v(i)[0] \in [B]^d$ (just write $i$ in base $B$),
then the node at index $i$ at layer $k \in [d]$ has an out-going edge
to each node $j$ at layer $k+1$ for which $v(j)$ and $v(i)$ differ
only in the $k$'th coordinate. Here the $0$'th coordinate is the
coordinate corresponding to the least significant digit when thinking
of $v(i)$ and $v(j)$ as numbers written in base $B$ (again see
Figure~\ref{fig:butterfly}). Observe that there is precisely one
directed path between each source-sink pair. For the $s$'th source and
the $t$'th sink, this path corresponds to ``morphing'' one digit of
$v(s)$ into the corresponding digit in $v(t)$ for each layer traversed
in the butterfly graph.

The input to the problem of reachability oracles in the butterfly
graph, with degree $B$ and depth $d$, is a subset of the edges of the
butterfly graph, i.e.\ we are given a subgraph~$G$ of the butterfly as
input. A query is specified by a source-sink pair~$(s,t)$ and the goal
is to return whether there exists a directed path from the given
source~$t$ to the given sink~$t$ in $G$.

\begin{figure}
  \centering
  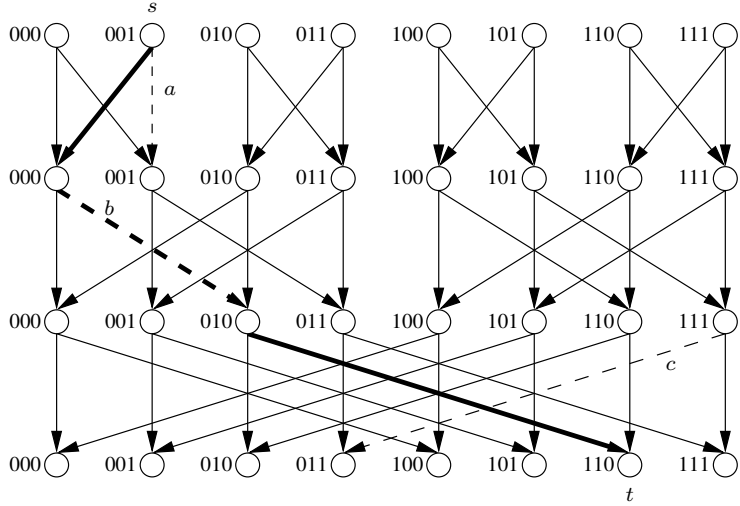  
  \caption{A butterfly with degree $B=2$ and depth $d=3$. The path
    shown in \textbf{bold} is the unique path from the source~$s=001$
    to the sink~$t=110$. A concrete input to the \emph{reachability
      oracles in the butterfly graph} problem consists of a subset of
    the edges of the butterfly. An example input is obtained by
    deleting the dashed edges labelled $a,b$ and $c$. For that input,
    there is no path from the source~$s$ to the sink~$t$ since the
    edge $b$ is not part of the input.}
  \label{fig:butterfly}
\end{figure}

P{\v a}tra{\c s}cu proved the following lower bound for reachability
oracles:

\begin{theorem}[P{\v a}tra{\c s}cu~\cite{patrascu11structures}, Section 5]
\label{thm:butterfly}
  Any cell probe data structure answering reachability queries in
  subgraphs of the butterfly graph with degree $B$ and depth $d$,
  having space usage $S$ words of $w$ 
  bits, must have
  query time $t=\Omega(d)$, provided $B = \Omega(w^2)$ and $\lg B =
  \Omega(\lg Sd/N)$. Here $N$ denotes the number of non-sink nodes in
  the butterfly graph.
\end{theorem}

We derive our lower bound by showing that any cell probe data
structure for skyline range counting can be used to answer
reachability queries in subgraphs of the butterfly graph for any
degree $B$ and depth $d$.

\paragraph{Edges to 2-d Rectangles.}

Consider the butterfly graph with degree $B$ and depth $d$. The first
step of our reduction is inspired by the reduction P{\v a}tra{\c s}cu
used to obtain a lower bound for $2$-d rectangle stabbing: Consider an
edge of the butterfly graph, leaving the $i$'th node at layer $k \in
[d]$ and entering the $j$'th node in layer $k+1$. We denote this edge
$e_k(i,j)$. The source-sink pairs~$(s,t)$ that are connected through
$e_k(i,j)$ are those for which:
\begin{enumerate}
\item The source has an index $s$ satisfying $v(s)[h] = v(i)[h]$ for
  $h \geq k$, i.e.\ $s$ and $i$ agree on the $d-k$ most significant
  digits when written in base $B$.
\item The sink has an index $t$ satisfying $v(t)[h] = v(j)[h]$ for $h
  \leq k+1$, i.e.\ $t$ and $j$ agree on the $k+1$ least significant
  digits when written in base $B$.
\end{enumerate}
We now map each edge $e_k(i,j)$ of the butterfly graph to a rectangle
in $2$-d. For the edge $e_k(i,j)$, we create the rectangle
$r_k(i,j)=[x_1,x_2] \times [y_1,y_2]$ where:
\begin{itemize}
\item $x_1 = v(i)[d-1]v(i)[d-2]\cdots v(i)[k]0\cdots 0$ when
  written in base $B$,
\item $x_2 = v(i)[d-1]v(i)[d-2]\cdots v(i)[k](B-1)\cdots (B-1)$ when
  written in base $B$,
\item $y_1 = v(j)[0]v(j)[1]\cdots v(j)[k+1]0 \cdots 0$ when written in
  base $B$, and
\item $y_2 = v(j)[0]v(j)[1]\cdots v(j)[k+1](B-1) \cdots (B-1)$ when written in
  base $B$.
\end{itemize}
The crucial observation is that for a source-sink pair, where the
source is the $s$'th source and the sink is the $t$'th sink, the edges
on the path from the source to the sink in the butterfly graph are
precisely those edges $e_k(i,j)$ for which the corresponding
rectangle $r_k(i,j)$ contains the point $(s,\rev_B(t))$, where
$\rev_B(t)$ is the number obtained by writing $t$ in base $B$ and then
reversing the digits.

We now collect the set of rectangles $\Rects$, containing each
rectangle $r_k(i,j)$ corresponding to an edge of the butterfly
graph. Given an input subgraph $G$, we \emph{mark} all rectangles
$r_k(i,j) \in \Rects$ for which the corresponding edge $e_k(i,j)$ is
also in $G$. It follows that there is a directed path from the $s$'th
source to the $t$'th sink in the subgraph $G$ if and only if
$(s,\rev_B(t))$ is not contained in any \emph{unmarked} rectangle in
$\Rects$.

Our goal is now to transform marked and unmarked rectangles to points,
such that we can use a skyline counting data structure to determine
whether a given point $(s,\rev_B(t))$ is contained in an unmarked
rectangle. Note that our reduction only works for the rectangle set
$\Rects$ obtained from the butterfly graph, and not for any set of
rectangles, i.e.\ we could not have reduced from the general problem
of 2-d rectangle stabbing.

\paragraph{2-d Rectangles to Points.}
To avoid tedious details, we from this point on allow the input to
skyline queries to have multiple points with the same $x$- or
$y$-coordinate (though not two points with both coordinates
identical). This assumption can easily be removed, but it would only
distract the reader from the main ideas of our reduction. We still use
the definition that a point $(x,y)$ dominates a point $(x',y')$ if and
only if $x' \leq x$ and $y' \leq y$.

The next step of the reduction is to map the rectangles $\Rects$ to a
set of points. For this, we first transform the coordinates slightly:
For every rectangle $r_k(i,j) \in \Rects$, having coordinates $[x_1 ,
  x_2 ] \times [y_1 , y_2]$, we modify each of the coordinates in the
following way:
\begin{itemize}
\item $x_1 \gets dx_1+(d-1-k)$,
\item $x_2 \gets dx_2+d-1$,
\item $y_1 \gets dy_1+k$, and
\item $y_2 \gets dy_2+d-1$.
\end{itemize}
The multiplication with $d$ essentially corresponds to expanding each
point with integer coordinates to a $d \times d$ grid of points. The
purpose of adding $k$ to $y_1$ and $(d-1-k)$ to $x_1$ is to ensure
that, if two rectangles share a lower-left corner (only possible for
two rectangles $r_k(i,j)$ and $r_{k'}(i',j')$ where $k \neq k'$), then
those corners do not dominate each other in the transformed set of
rectangles. We will see later that the particular placement of the
points based on $k$ also plays a key role. We use $\pi : [B^d]^4 \to
[dB^d]^4$ to denote the above map. With this notation, the transformed
set of rectangles is denoted $\pi(\Rects)$ and each rectangle
$r_k(i,j) \in \Rects$ is mapped to $\pi(r_k(i,j)) \in \pi(\Rects)$.

We now create the set of points $P'$ containing the set of lower-left
corner points for all rectangles $\pi(r_k(i,j)) \in \pi(\Rects)$,
i.e.\ for each $\pi(r_k(i,j))=[x_1,x_2] \times [y_1, y_2]$, we add the
point $(x_1,y_1)$ to $P'$. See Figure~\ref{fig:rects} for an
example. The set $P'$ has the following crucial property:

\begin{figure}[t]
  \centering
  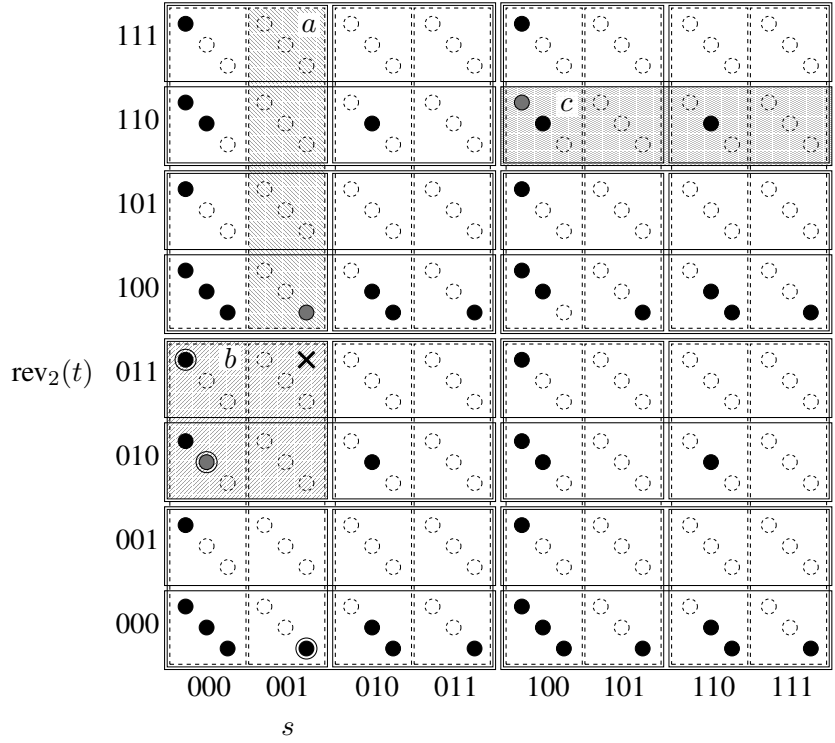  
  \caption{The butterfly with degree $B=2$ and depth $d=3$ from
    Figure~\ref{fig:butterfly} translated to a set of rectangles. The
    dashed rectangles correspond to the edges $a,b$ and $c$ from
    Figure~\ref{fig:butterfly}. Every grid point is replaced by up to
    $d$ points placed on a diagonal. Each rectangle obtained from an
    edge of the butterfly graph produces one point on the diagonal
    corresponding to the rectangle's lower left corner. The points
    corresponding to the rectangles obtained from edges $a,b$ and $c$
    are shown in gray. The query corresponding to the source~$s=001$
    and the sink~$t=110$ in Figure~\ref{fig:butterfly} is translated
    to the two-sided skyline query rectangle with its upper right
    corner at the $\times$. The double circled points are the points
    on the skyline of the query range and these correspond exactly to
    the lower left corners of the rectangles containing the $\times$.}
  \label{fig:rects}
\end{figure}

\begin{lemma}
\label{lem:sky}
Let $(x,y)$ be a point with coordinates in $[B^d] \times [B^d]$. Then
for the two-sided query rectangle $Q =(-\infty,dx+d-1] \times
  (-\infty,dy+d-1]$, it holds that $\Skyline(Q \cap P')$ contains
    precisely the points in~$P'$ corresponding to the lower-left
    corners of the rectangles $\pi(r_k(i,j)) \in \pi(\Rects)$ for which
    $r_k(i,j)$ contains~$(x,y)$.
\end{lemma}

\begin{proof}
First let $p=(x_1,y_1) \in P'$ be the lower-left corner of a rectangle
$\pi(r_k(i,j))$ such that $r_k(i,j)$ contains the point $(x,y)$. We
want to show that $p \in \Skyline(Q \cap P')$. Since $r_k(i,j)$
contains the point $(x,y)$, we have $x \geq \lfloor x_1/d \rfloor$ and
$y \geq \lfloor y_1/d\rfloor$. From this, we get $dx+d-1 \geq d\lfloor
x_1/d \rfloor + (d-1-k) = x_1$ and $dy+d-1 \geq d \lfloor y_1/d
\rfloor + k = y_1$, i.e.\ $p$ is inside $Q$. Since $(x,y)$ is inside
$r_k(i,j)$, we also have that $(dx+d-1,dy+d-1)$ is dominated by the
upper-right corner of $\pi(r_k(i,j))$, i.e.\ $(dx+d-1,dy+d-1)$ is
inside $\pi(r_k(i,j))$.

What remains to be shown is that no other point in $Q \cap P'$
dominates $p$. For this, assume for contradiction that some point
$p'=(x_1',y_1') \in P'$ is both in $Q$ and also dominates $p$. First,
since $p'$ is dominated by $(dx+d-1,dy+d-1)$ and also dominates $p$,
we know that $p'$ must be inside $\pi(r_k(i,j))$. Now let
$\pi(r_{k'}(i',j')) \neq \pi(r_k(i,j))$ be the rectangle in
$\pi(\Rects)$ from which $p'$ was generated, i.e.\ $p'$ is the
lower-left corner of $\pi(r_{k'}(i',j'))$. We have three cases:
\begin{enumerate}
\item First, if $k'=k$ we immediately get a contradiction since the
  rectangles $\pi(\Rects)_k = \{\pi(r_{k'}(i',j')) \in \pi(\Rects)
  \mid k' = k\}$ are pairwise disjoint and hence $p'$ could not have
  been inside $\pi(r_k(i,j))$.
\item If $k' < k$, we know that $\pi(r_{k'}(i',j'))$ is shorter in
  $x$-direction and longer in $y$-direction than $\pi(r_k(i,j))$. From
  our transformation, we know that $(y_1 \bmod d) = k$ and $(y'_1 \bmod
  d) = k' < k$. Thus since $p'$ dominates $p$, we must have $\lfloor
  y_1' / d \rfloor > \lfloor y_1/d \rfloor$. But these two values are
  precisely the $y$-coordinates of the lower-left corners of
  $r_k(i,j)$ and $r_{k'}(i',j')$. By definition, we get:
\[
v(j')[0]v(j')[1]\cdots v(j')[k'+1]0 \cdots 0 > v(j)[0]v(j)[1] \cdots
  v(j)[k+1]0 \cdots 0\;.
\]
Since $k'<k$, this furthermore gives us
\[
v(j')[0]v(j')[1]\cdots v(j')[k'+1] > v(j)[0]v(j)[1] \cdots
  v(j)[k'+1]\;.
\]
From this it follows that
\[
v(j')[0]v(j')[1]\cdots v(j')[k'+1]0 \cdots 0 > v(j)[0]v(j)[1] \cdots
  v(j)[k+1](B-1) \cdots (B-1)\;,
\]
i.e.\ the lower-left corner of $r_{k'}(i',j')$ is outside $r_k(i,j)$,
which also implies that the lower-left corner of $\pi(r_{k'}(i',j'))$
is outside $\pi(r_{k}(i,j))$. That is, $p'$ is outside $\pi(r_k(i,j))$,
which gives the contradiction.

\item The case for $k'>k$ is symmetric to the case $k'<k$, just using
  the $x$-coordinates instead of the $y$-coordinates to derive the
  contradiction.
\end{enumerate}
The last step of the proof is to show that no point $p=(x_1,y_1) \in
P'$ can be in $\Skyline(Q \cap P')$ but at the same time correspond to
the lower-left corner of a rectangle $\pi(r_k(i,j))$ where $r_k(i,j)$
does not contains the point $(x,y)$. First observe that
$(dx+d-1,dy+d-1)$ is contained in precisely one rectangle
$\pi(r_{k'}(i',j'))$ for each value of $k' \in [d]$. Now let
$\pi(r_{k}(i',j')) \neq \pi(r_{k}(i,j))$ be the rectangle containing
$(dx+d-1,dy+d-1)$ amongst the rectangles $\pi(\Rects)_k$. The
lower-left corner of this rectangle is dominated by $(dx+d-1,dy+d-1)$
but also dominates $p$, hence $p$ is not in $\Skyline(Q \cap P')$.
\end{proof}

\begin{figure}[t]
  \centering
  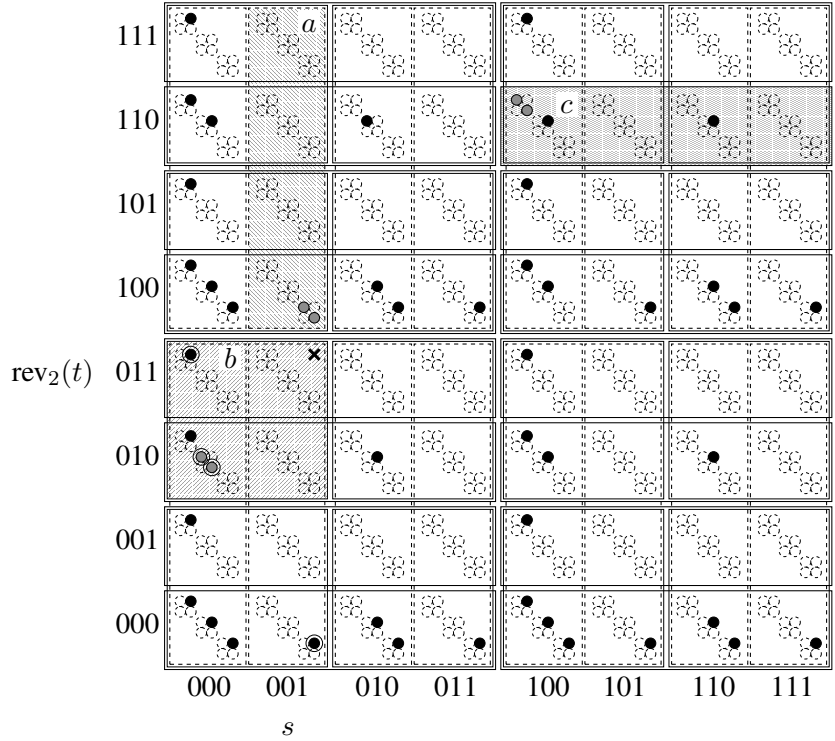  
  \caption{Points corresponding to unmarked rectangles are replaced by
    two points. The example from Figure~\ref{fig:butterfly} and
    Figure~\ref{fig:rects} has three unmarked rectangles,
    corresponding to edges $a,b$ and $c$ of the butterfly graph. As
    shown, these rectangles become two (gray) input points and marked
    rectangles are represented by only one input point. The upper
    right corner of the two-sided query rectangle corresponding to the
     source~$s=001$ and sink~$t=110$ in the previous examples is
    shown as a $\times$. The double circled points are the points on
    the skyline of the query range. As can be seen, the unmarked
    rectangle corresponding to the edge labelled $b$ contributes two
    points to the skyline of the query $\times$.}
  \label{fig:rects2}
\end{figure}

\paragraph{Handling Marked and Unmarked Rectangles.}
The above steps are all independent of the concrete input subgraph
$G$. As discussed, we need a way to determine whether a query point is
contained in an unmarked rectangle or not. This step is now very
simple in light of Lemma~\ref{lem:sky}: First, multiply all
coordinates of points in $P'$ by $2$. This corresponds to expanding
each point with integer coordinates into a $2 \times 2$ grid. Now for
every point $p \in P'$, if the rectangle $\pi(r_k(i,j))$ from which
$p$ was generated is marked, then we add $1$ to both the $x$- and
$y$-coordinate of $p$, i.e.\ we move $p$ to the upper-right corner of
the $2 \times 2$ grid in which it is placed. If $\pi(r_k(i,j))$ is
unmarked, we replace it by two points, one where we add $1$ to the
$x$-coordinate, and one where we add $1$ to the $y$-coordinate, see
Figure~\ref{fig:rects2}. We denote the resulting set of points $P(G)$. It
follows immediately that:

\begin{corollary}
  Let $G$ be a subgraph of the butterfly graph with degree $B$ and
  depth~$d$. Also, let $(x,y)$ be a point with coordinates in $[B^d]
  \times [B^d]$. Then for the two-sided query rectangle $Q
  =(-\infty,2d(x+1)-1] \times (-\infty,2d(y+1)-1]$, it holds that
      $\Skyline(Q \cap P(G))$ contains precisely one point from $P(G)$
      for every marked rectangle in $\Rects$ that contains $(x,y)$,
      two points from $P(G)$ for every unmarked rectangle in $\Rects$
      that contains $(x,y)$, and no other points, i.e.\
      $|\Skyline(Q \cap P(G))|-d$ equals the number of unmarked
      rectangles in $\Rects$ which contains $(x,y)$.
\end{corollary}

\begin{corollary}
\label{cor:reduct}
  Let $G$ be a subgraph of the butterfly graph with degree $B$ and
  depth~$d$. Let $s$ be the index of a source and $t$ the index of a
  sink. Then the $s$'th source can reach the $t$'th sink in $G$ 
  if and only if
  $|\Skyline(Q \cap P(G))|=d$ for the two-sided query rectangle $Q =
  (-\infty,2d(s+1)-1] \times (-\infty,2d(\rev_B(t)+1)-1]$.
\end{corollary}

\paragraph{Deriving the Lower Bound.} 

The lower bound can be derived from Corollary~\ref{cor:reduct} and
Theorem~\ref{thm:butterfly} as follows.  First note that the set
$\Rects$ contains $NB$ rectangles, since each rectangle corresponds to
an edge of the buttefly graph and each of the $N$ non-sink nodes of
the butterfly graph has $B$ outgoing edges. Each of these rectangles
gives one or two points in~$P(G)$. Letting $n$ denote $|P(G)|$, we
have $NB\leq n\leq 2NB$. From $N=d\cdot B^d\le n$ we
get $d\le\lg n$ and $d=\Theta(\lg_B N)$.

Given $n$, $w\geq \lg n$, and $S\geq n$, we now derive a lower bound
on the query time.  Setting $B=\frac{S}{n}w^2$ we have $B=\Omega(w^2)$
and $\lg B=\Omega(\lg\frac{Sd}{N})$ (as required by
Theorem~\ref{thm:butterfly}), where the last bound follows from
$\lg\frac{Sd}{N}
\leq\lg\frac{S\cdot\lg n}{n/2B}
\leq\lg(2B\frac{S\cdot w}{n})
\leq\lg (2B^2)
=O(\lg B)$. 
Furthermore we have $\lg\frac{Sw}{n}
=\frac{1}{2}\lg(\frac{Sw}{n})^2
\geq\frac{1}{2}\lg(\frac{S}{n}w^2)
=\frac{1}{2}\lg B$. 
From Theorem~\ref{thm:butterfly} we can now bound the time for a 
skyline counting query by 
$t 
= \Omega(d) 
= \Omega(\lg_B N)
= \Omega(\lg n/\lg B)
= \Omega(\lg n/\lg(Sw/n))$.

\section{Skyline Counting Data Structure}
\label{sec:upper-bound}

In this section we describe a data structure using $O(n)$ space
supporting orthogonal skyline counting queries in $O(\lg n/\lg\lg n)$
time. We first describe the basic idea of how to support queries, then
present the details of the stored data structure and the details of the
query.

The basic idea is to store the $n$~points in left-to-right $x$-order
at the leaves of a balanced tree~$T$ of degree
$\Theta(\log^{\varepsilon} n)$, i.e. height~$O(\log n/\log\log n)$,
and for each internal node $v$ have a list $L_v$ of the points in the
subtree rooted in~$v$ in sorted $y$-order.  The \emph{slab} of $v$ is
the narrowest infinite vertical band containing $L_v$. To obtain the
overall linear space bound, $L_v$ will not be stored explicitly but
implicitly and rank-reduced using rank-select data structures, where
navigation is performed using fractional cascading on rank-select data
structures (details below). A 4-sided query $R$ decomposes into
2-sided subqueries at $O(\log n/\log\log n)$ nodes (in
Figure~\ref{fig:base-tree}, $R$ is decomposed into subqueries
$R_1$-$R_5$, white points are nodes on the skyline within $R$, double
circled points are the topmost points within each $R_i$). For skyline
queries (both counting and reporting) it is important to consider the
subqueries right-to-left, and the lower $y$-value for the subquery in
$R_i$ is raised to the maximal $y$-value of a point in the subqueries
to the right. Since the tree~$T$ has non-constant degree, we need
space efficient solutions for multislab queries at each node~$v$. We
partition $L_v$ into blocks of size $O(\log^{2\varepsilon} n)$, and a
query $R_i$ decomposes into five subqueries (1-5), see
Figure~\ref{fig:multicolumn}: (1) and (3) are on small subsets of
points within a single block and can be answered by tabulation (given
the \emph{signature} of the block); (2) is a block aligned multislab
query; (4) and (5) are for single slabs (at the children of $v$).  For
(2,4,5) the skyline size between points $i$ and $j$ (numbered
bottom-up) can be computed as one plus the difference between the size
of the skyline from 1 to $j$ and 1 to $k$, where $k$ is the rightmost
point between $i$ and $j$ (see Figure~\ref{fig:skyline-diff}, white
and black circles and crosses are all points, crosses indicate the
skyline from $i$ to $j$, white circles from 1 to $k$, and white
circles together with crosses from 1 to $j$). Finally, the skyline
size from 1 to $i$ can be computed from a prefix sum, if we for point
$i$ store the number of points in the skyline from 1 to $i-1$
dominated by $i$ (see Figure~\ref{fig:skyline-sum}, the skyline
between 1 and 6 consists of the three white nodes, and the size is
$6-(2+0+0+0+1+0)=3$).

We let $\Delta=\max\{2,\ceil{\lg^\varepsilon n}\}$ be a parameter of
our construction, where $0<\varepsilon<1/3$ is a constant.  We build a
balanced \emph{base tree} $T$ over the set of points $P$, where the
leafs from left-to-right store the points in $P$ in sorted order
w.r.t.\ $x$-coordinate. Each internal node of $T$ has degree at
most $\Delta$ and $T$ has height $\ceil{\lg_\Delta n}+1$. See
Figure~\ref{fig:base-tree}.

\begin{figure}[t]
  \centering
  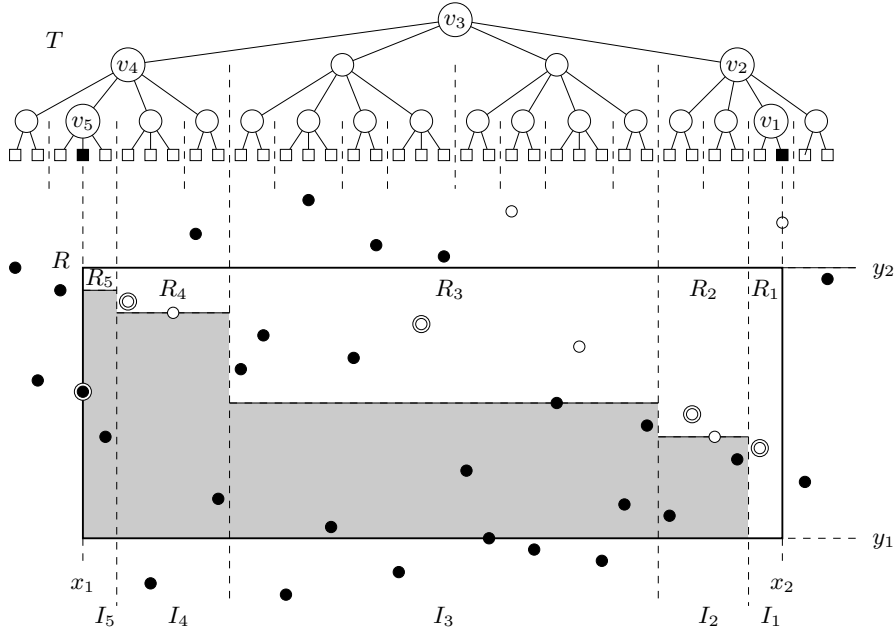  
  \caption{The base tree $T$ with $\Delta=4$, and the decomposition of
    a query into a sequence of multislab queries $R_1$-$R_5$. White
    points are nodes on the skyline within $R$. The double circled
    points are the topmost points within each of the multislabs.}
  \label{fig:base-tree}
\end{figure}

For each internal node $v$ of $T$ we store a set of data structures.
Before describing these we need to introduce some notation. The
subtree of $T$ rooted at a node $v$ is denoted $T_v$, and the set of
points stored at the leaves of $T_v$ is denoted $P_v$. We let
$n_v=|P_v|$ and $L_v[1..n_v]$ be the list of the points in $P_v$
sorted in increasing $y$-order. We let $I_v=[\ell_v,r_v]$ denote the
$x$-interval defined by the $x$-coordinates of the points stored at
the leaves of~$T_v$, and denote $I_v \times [n]$ the \emph{slab}
spanned by $v$.  The degree of $v$ is denoted $d_v$, the children of
$v$ are from left-to-right denoted $c_v^1,\ldots,c_v^{d_v}$, and the
parent of node $v$ is denoted $p_v$.  A list~$L_v$ is partitioned into
a sequence of blocks $B_v[1..\ceil{n_v/\Delta^2}]$ of size~$\Delta^2$,
such that $B_v[i]=L_v[(i-1)\Delta^2+1..\min\{n_v,i\Delta^2\}]$.  The
\emph{signature} $\sigma_v[i]$ of a block $B_v[i]$ is a list of pairs:
For each point~$p$ from $B_v[i]$ in increasing $y$-order we construct
a pair~$(j,r)$, where $j$ is the index of the child $c_v^j$ of $v$
storing $p$ and $r$ is the rank of $p$'s $x$-coordinate among all
points in $B_v[i]$ stored at the same child $c_v^j$ as $p$.  The total
number of bits required for a signature is at most $\Delta^2(\lg
\Delta+\lg \Delta^2)=O(\lg^{2\varepsilon} n\cdot\lg\lg n)$.

To achieve overall $O(n)$ space we need to encode succinctly
sufficient information for performing queries.  In particular we will
\emph{not} store the points in $L_v$ explicitly at the node $v$, but
only partial information about the points relative position will be
stored.

Queries on a block $B_v[i]$ are handled using table lookups in global
tables using the block signature $\sigma_v[i]$.  We have tables for
the below block queries, where we assume $\sigma$ is the signature of
a block storing points $p_1,\ldots,p_{\Delta^2}$ distributed in
$\Delta$ child slabs.

\begin{description}
\item[$\mathrm{Below}(\sigma,t,i)$] Returns the number of points from
  $p_1,\ldots,p_t$ contained in slab $i$.

\item[$\mathrm{Rightmost}(\sigma,b,t,i,j)$] Returns $k$, where $p_k$
  is the rightmost point among $p_b,\ldots,p_t$ contained in slabs
  $[i,j]$. If no such point exists, -1 is returned.

\item[$\mathrm{Topmost}(\sigma,b,t,i,j)$] Returns $k$, where $p_k$ is
  the topmost point among $p_b,\ldots,p_t$ contained in slabs
  $[i,j]$. If no such point exists, -1 is returned.

\item[$\mathrm{SkyCount}(\sigma,b,t,i,j)$] Returns the size of the
  skyline for the subset of the points $p_b,\ldots,p_t$ contained in
  slabs~$[i,j]$.

\end{description}

The arguments to each of the above lookups consists of at most
$|\sigma|+2\lg \Delta^2+2\lg \Delta=|\sigma|+O(\lg\lg
n)=O(\lg^{2\varepsilon} n\cdot\lg\lg n)$ bits and the answer is
$\lg(\Delta+1)=O(\lg\lg n)$ bits, i.e.\ each query can be answered in
$O(1)$ time using a table of size $O(2^{\lg^{2\varepsilon} n\cdot
  \lg\lg n}\cdot\lg\lg n)=o(n)$ bits, since $\varepsilon<1/3$.

For each internal node $v$ of $T$ we store the following data
structures, each having $O(1)$ access time.

\begin{description} 
\item[$C_v(i)$] Compact array that for each $i$, where $1\leq i\leq
  n_v$, stores the index of the child of $v$ storing $L_v[i]$,
  i.e.\ $1\leq C_v(i)\leq \Delta$.  Space usage $O(n_v\lg \Delta)$
  bits.

\item[{$\pi_v(i)$}] For each $i$, $1\leq i\leq n_v$, stores the index
  of $L_v[i]$ in $L_{p_v}$, i.e.\ $L_{p_v}[\pi_v(i)]=L_v[i]$. This can
  be supported by constructing the select data structure of
  Lemma~\ref{lem:rankselect} on the bit-vector $X$, where $X[i]=1$ if
  and only if $L_{p_v}[i]$ is in $L_v$. A query to $\pi_v(i)$ simply
  becomes a select$(i)$ query.  Space usage $O(n_v\lg
  (n_{p_v}/n_v))=O(n_v\lg \Delta)$ bits.

\item[{$\sigma_v(i)$}] Array of signatures for the blocks
  $B_v[1..\ceil{n_v/\Delta^2}]$. Space usage
  $O(n_v/\Delta^2\cdot\Delta^2\cdot\lg\Delta)=O(n_v\lg\Delta)$ bits.

\item[{$\Pred_v(t,i)$ / $\Succ_v(t,i)$}] Supports finding the
  predecessor/successor of $L_v[t]$ in the $i$'th child list
  $L_{c_v^i}$. Returns $\max \{ k \mid 1\leq k \leq n_{c_v^i} \wedge
  \pi_{c_v^i}[k]\leq t \}$ and $\min \{ k \mid 1\leq k \leq n_{c_v^i}
  \wedge \pi_{c_v^i}[k]\geq t \}$, respectively.  For each child
  index~$i$, we construct an array $X^i$ of size $\ceil{n/\Delta^2}$,
  such that $X^i[b]$ is the number of points in block $B_v[b]$ that
  are stored in the $i$'th child slab. The prefix sums of each $X^i$
  are stored using the data structure of Lemma~\ref{lem:prefixsum}
  using space $O((n_{v}/\Delta^2)\lg(\Delta^2))$ bits. The total
  space for all $\Delta$ children of $v$ becomes $O(\Delta\cdot
  n_v/\Delta^2 \cdot \lg\Delta)=O(n_v)$ bits. The result of a
  $\Pred_v(t,i)$ query is $\sum_{j=1}^{\ceil{t/\Delta^2}-1}
  X^i[j]+\mathrm{Below}(\sigma_v(\ceil{t/\Delta^2}),1+(t-1\bmod\Delta^2),i)$,
  where the first term can be computed in $O(1)$ time by
  Lemma~\ref{lem:prefixsum} and the second term is a constant time
  global table lookup. The result of $\Succ_v(t,i)=\Pred_v(t,i)$ if
  $C_v[t]=i$, otherwise $\Succ_v(t,i)=\Pred_v(t,i)+1$.

\item[{$\mathrm{Rightmost}_v(i,j)$}] Returns the index $k$, where
  $i\leq k\leq j$, such that $L_v[k]$ has the maximum $x$-value among
  $L_v[i..j]$. Using Lemma~\ref{lem:rmq} on the array of the
  $x$-coordinates of the points in $L_v$ we achieve $O(1)$ time
  queries and space usage $O(n_v)$ bits.

\item[{$\mathrm{SkyCount}_v(i)$}] Returns
  $|\Skyline(L_v[1..i])|$. Construct an array $X$, where $X[i]$ is the
  number of points in $\Skyline(L_v[1..i-1])$ dominated by
  $L_v[i]$. See Figure~\ref{fig:skyline-sum}. We can now compute
  $|\Skyline(L_v[1..i])|$ as $i-\sum_{j=1}^i X[j]$.  Using
  Lemma~\ref{lem:prefixsum} the query time becomes $O(1)$ and the
  space usage $O(n_v)$ bits, since $\sum_{j=1}^{n_v} X[j]\leq n_v-1$.

\item[{$\mathrm{SkyCount}_v(i,j)$}] Returns $|\Skyline(L_v[i..j])|$,
  computable by the following expression (see
  Figure~\ref{fig:skyline-diff}):
  $$\mathrm{SkyCount}_v(j)-\mathrm{SkyCount}_v(\mathrm{Rightmost}_v(i,j))+1\;.$$

\end{description}

\begin{figure}
  \centering
  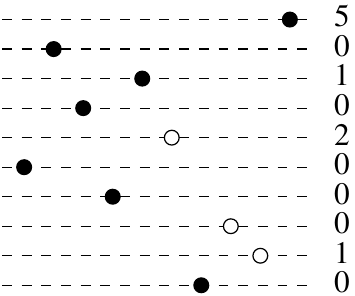  
  \caption{Computation of $|\Skyline(L_v[1..i])|$. To the right of
    each point $L_v[i]$ is shown the number of points in
    $\Skyline(L_v[1..i-1])$ dominated by $L_v[i]$. The skyline of
    $L_v[1..6]$ consists of the three white
    nodes. $|\Skyline(L_v[1..6])|=6-2-0-0-0-1-0=3$.}
  \label{fig:skyline-sum}
\end{figure}

\begin{figure}
  \centering
  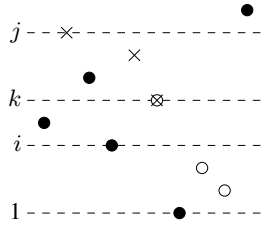  
  \caption{Illustration of $\mathrm{SkyCount}_v(i,j)$.  White and black
    circles and crosses are all points. $L_v[k]$ is the rightmost
    point in $L_v[i..j]$. Crosses indicate $\Skyline(L_v[i..j])$,
    white circles indicate $\Skyline(L_v[1..k])$, and white circles
    together with crosses is $\Skyline(L_v[1..j])$.}
  \label{fig:skyline-diff}
\end{figure}

Finally, we store for each node $v$ and slab interval $[i,j]$ the
following data structures.

\begin{description}

\item[$\mathrm{Rightmost}_{v,i,j}(b,t)$] Returns $k$, where $L_v[k]$
  is the rightmost point among the points in blocks $B_v[b..t]$
  contained in slabs $[i,j]$. If no such point exists, -1 is
  returned. Can be solved by applying Lemma~\ref{lem:rmq} to the array
  $X$, where $X[s]$ is the $x$-coordinate of the rightmost point in
  $B_v[s]$ contained in slabs $[i,j]$. A query first finds the block
  $\ell$ containing the rightmost point using this data structure, and
  then returns
  $(\ell-1)\Delta^2+\mathrm{Rightmost}(\sigma_v[\ell],1,\Delta^2,i,j)$.
  Space usage $O(n_v/\Delta^2)$ bits.

\item[$\mathrm{Topmost}_{v,i,j}(b,t)$] Returns $k$, where $L_v[k]$ is
  the topmost point among the points in blocks $B_v[b..t]$ contained
  in slabs $[i,j]$. If no such point exists, -1 is returned. Can be
  solved by first using Lemma~\ref{lem:rmq} on the array $X$, where
  $X[s]=s$ if there exists a point in $B_v[s]$ contained in slabs
  $[i,j]$. Otherwise $X[s]=0$. Let $\ell$ be the block found using
  Lemma~\ref{lem:rmq}. Return the result of
  $(\ell-1)\Delta^2+\mathrm{Topmost}(\sigma_v[\ell],1,\Delta^2,i,j)$. 
  Space usage $O(n_v/\Delta^2)$ bits.

\item[$\mathrm{SkyCount}_{v,i,j}(b,t)$] Returns the size of the
  skyline for the subset of points in blocks $B_v[b..t]$ contained in
  slabs $[i,j]$. Can be supported by two applications of
  Lemma~\ref{lem:prefixsum} on two arrays $X$ and $Y$ as follows.  Let
  $X[s]=\mathrm{SkyCount}(\sigma_v[s],1,\Delta^2,i,j)$, i.e.\ the size
  of the skyline of the points in block $B_v[s]$ contained in slabs
  $[i,j]$.  Let $B_{v,i,j}[s]$ denote the points in $B_v[s]$ contained
  in slabs~$[i,j]$.  Let $Y[s]=|\Skyline(B_{v,i,j}[1..s-1]) \setminus
  \Skyline(B_{v,i,j}[1..s])|$, i.e.\ the number of points on
  $\Skyline(B_{v,i,j}[1..s-1])$ dominated by points in $B_{v,i,j}[s]$.
  Space usage for $X$ and $Y$ is $O(n_v/\Delta^2\cdot\lg \Delta^2)$
  bits.  We can compute $\mathrm{SkyCount}_{v,i,j}(b,t)=\sum_{s=k}^t
  X[s]-\sum_{s=k+1}^t Y[s]$, where
  $k=\ceil{\mathrm{Rightmost}_{v,i,j}(b,t)/\Delta^2}$.

\end{description}

The total space of our data structure, in addition to the $o(n)$ bits
for our global tables, can be bounded as follows.  The total space for
all $O(\Delta^2)$ multislab data structures for a node $v$ is
$O(\Delta^2\cdot n_v/\Delta^2\cdot\lg \Delta)$ bits. The total space
for all data structures at a node $v$ becomes $O(n_v\lg \Delta)$
bits. Since the sum of all $n_v$ for a level of $T$ is at most $n$,
the total space for all nodes at a level of $T$ is $O(n\lg\Delta)$
bits. Since $T$ has height~$O(\lg_\Delta n)$, the total space usage
becomes $O(n\lg\Delta\cdot\lg_\Delta n)=O(n\lg n)$ bits,
i.e.\ $O(n)$ words.
The data structure can be constructed bottom-up in $O(n\log
n)$ time.

\subsection{Skyline Range Counting Queries}
\label{sec:query}

To answer a skyline counting query $R=[x_1,x_2]\times[y_1,y_2]$, we
identify the nodes on the paths in $T$ from the two leaves storing
$x_1$ and $x_2$ up to the lowest common ancestor of the two
leaves. Let $v_1,\ldots,v_m$ be the set of these nodes in a
right-to-left traversal in $T$ (see Figure~\ref{fig:base-tree}). The
horizontal span of the query, $[x_1,x_2]$, is the concatenation of the
span of at most one multislab $I_1,\ldots,I_m$ from each of
$v_1,\ldots,v_m$. For each such multislab $I_\ell$ we form a new
subquery $R_\ell=I_\ell\times [z_\ell,y_2]$, completely spanning the
multislab in the horizontal direction and vertically has a range
$[z_\ell,y_2]$, where $z_1=y_1$ and
$z_\ell=\max\{z_{\ell-1},y_\ell^{\max}+1\}$, for $\ell=2$ to $m$ and
$y_\ell^{\max}$ is the maximal $y$-coordinate of a point in
$I_{\ell-1}\times[1,y_2]$. By definition of the $R_\ell$ queries, the
skyline of the points contained within $R$ is exactly the union of the
skylines for each of the $R_\ell$ subqueries (see
Figure~\ref{fig:base-tree}), since the points in $R_\ell$ cannot be
dominated by other points that are both in $R$ and to the right of
$I_\ell$.

To navigate in $T$ we need to find the index of the successor of $y_1$
and the predecessor $y_2$ in each of the $L_{v_\ell}$ lists.  We start
with $y_1$ and $y_2$ being the indexes at the root, and then use the
$\mathrm{Succ}_v$/$\mathrm{Pred}_v$ structures at the nodes to find
the successor of $y_1$ and predecessor of $y_2$ at all the nodes on
the two paths from the root to $x_1$ and $x_2$. To find the topmost
point below $y_2$ in a multislab we use $\mathrm{Topmost}_{v,i,j}$. To
navigate $y^{\max}$ values up and down between the levels of $T$ we
use $\pi_v(y^{\max})$ to move upwards and
$\mathrm{Succ}_v(y^{\max},j)$ to move downwards to a slab~$j$. These
navigations can be performed in $O(1)$ time per node on the paths,
i.e.\ total time $O(\lg_\Delta n)$.

What remains is to compute in $O(1)$ time the size the skyline within
a query range $R_\ell$. In the following we consider a query range that
horizontally spans the child slabs $[i,j]$ of a node $v$, and
vertically spans the indexes $[y_{\mathrm{bottom}},y_{\mathrm{top}}]$
in $L_v$.

If the query range is within a single block of $L_v$
(i.e.\ $\ceil{y_{\mathrm{bottom}}/\Delta^2}=\ceil{y_{\mathrm{top}}/\Delta^2}$),
we compute the skyline size as
\[
\mathrm{SkyCount}(\sigma_v(\ceil{y_{\mathrm{top}}/\Delta^2}),
  1+(y_{\mathrm{bottom}}-1 \bmod \Delta^2),
  1+(y_{\mathrm{top}}-1 \bmod\Delta^2),i,j)\;.
\]

Otherwise we decompose the skyline counting query into five subranges
(1)-(5), see Figure~\ref{fig:multicolumn}.  We first compute the
$y$-coordinate of the rightmost point~$p_1$ in the top block
$B_{\mathrm{top}}$ of the query range using
$$p_1.y = \mathrm{Rightmost}(\sigma_v(\ceil{y_{\mathrm{top}}/\Delta^2}),
1,1+(y_{\mathrm{top}}-1 \bmod\Delta^2),i,j)+\Delta^2\ceil{y_{\mathrm{top}}/\Delta^2-1}\;,$$
and compute the size of the skyline of the intersection of $B_{\ttop}$
and the query region by:
\begin{equation}
\mathrm{SkyCount}(\sigma_v(\ceil{y_{\mathrm{top}}/\Delta^2}),1,1+(y_{\mathrm{top}}-1
\bmod\Delta^2),i,j)\;.
\end{equation}
Let $k_1$ be the slab containing $p_1$, computable as
$k_1=C_v(p_1.y)$. If no point is found in block $B_{\mathrm{top}}$,
then $k_1=i-1$.

\begin{figure}
  \centering
  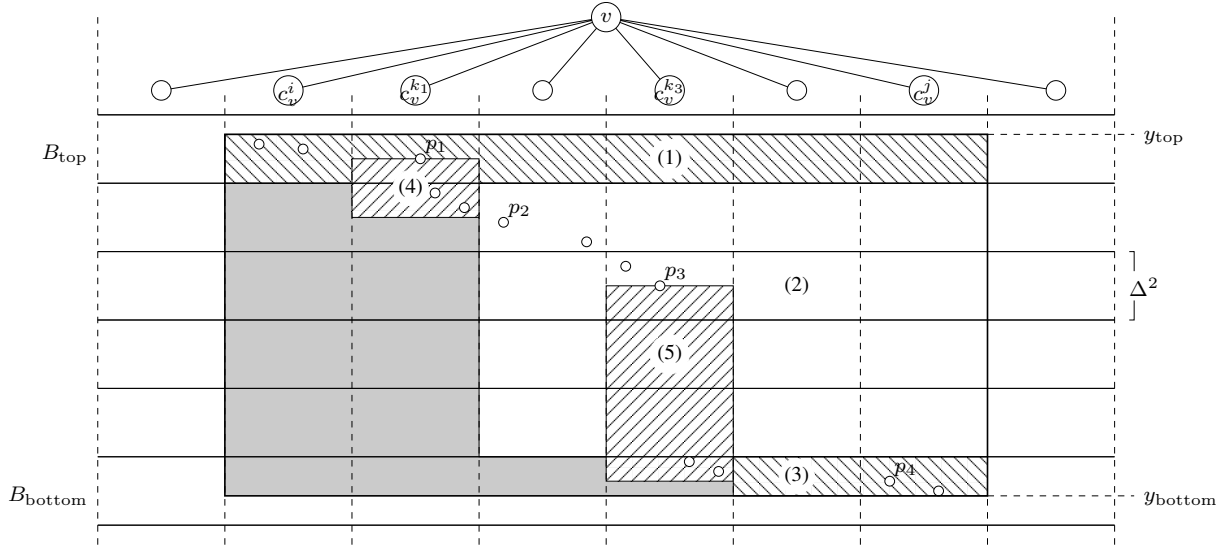  
  \caption{Skyline queries for multislabs.}
  \label{fig:multicolumn}
\end{figure}

Next we compute the $y$-coordinate of the topmost point $p_2$ in the
multislab query range spanning slabs $[k_1+1,j]$ and all blocks
between $B_{\mathrm{bottom}}$ and $B_{\mathrm{top}}$.
$$p_2.y=\mathrm{Topmost}_{v,k_1+1,j}(\ceil{y_{\mathrm{bottom}}/\Delta^2}+1,\ceil{y_{\mathrm{top}}/\Delta^2}-1)\;.$$
In the same subrange we find the $y$-coordinate of the rightmost point
$p_3$ using
$$p_3.y
=\mathrm{Rightmost}_{v,k_1+1,j}(\ceil{y_{\mathrm{bottom}}/\Delta^2}+1,\ceil{y_{\mathrm{top}}/\Delta^2}-1)\;.$$
Finally, the number of points on the skyline between $p_2$ and $p_3$
(including $p_2$ and $p_3$) is computed by
\begin{equation}
\mathrm{SkyCount}_{v,k_1+1,j}(\ceil{y_{\mathrm{bottom}}/\Delta^2}+1,\ceil{y_{\mathrm{top}}/\Delta^2}-1)\;.
\end{equation}

The slab containing the point $p_3$ is $k_3=C_v(p_3.y)$.  We
compute the number of points on the skyline to the right of $p_3$ in
block $B_{\mathrm{bottom}}$ by
\begin{equation}
\mathrm{SkyCount}(\sigma_v(\ceil{y_{\mathrm{bottom}}/\Delta^2}),1+(y_{\mathrm{bottom}}-1 \bmod\Delta^2),\Delta^2,k_3+1,j)\;,
\end{equation}
and the $y$-coordinate of the topmost point $p_4$ in block $B_{\mathrm{bottom}}$ contained in slabs $[k_3+1,j]$ by
$$p_4.y =
\mathrm{Topmost}(\sigma_v(\ceil{y_{\mathrm{bottom}}/\Delta^2}),1+(y_{\mathrm{bottom}}-1
\bmod\Delta^2),\Delta^2,k_3+1,j)+\Delta^2\ceil{y_{\mathrm{bottom}}/\Delta^2-1}\;.$$

The remaining points to be counted are the skyline points in slab
$k_1$ between $p_1$ and $p_2$, and in slab $k_3$ between $p_3$ and the
point $p_4$ in block $B_{\mathrm{bottom}}$. These values can
be computed by
\begin{equation}
\mathrm{SkyCount}_{c_v^{k_1}}(\mathrm{Succ}_v(p_2.y+1,k_1),\mathrm{Pred}_v(p_1.y,k_1))-1
\end{equation}
\begin{equation}
\mathrm{SkyCount}_{c_v^{k_3}}(\mathrm{Succ}_v(p_4.y,k_3),\mathrm{Pred}_v(p_3.y,k_3))-1\;,
\end{equation}
where we subtract one in both expressions, to avoid double counting $p_1$ and $p_3$.

Figure~\ref{fig:multicolumn} illustrates the five partial counts
computed. In the above we assumed that all queries ranges were
non-empty. In case $p_1$ does not exist, then $k_1=i-1$ and (4) is not
computed. If $p_4$ does not exist, then (5) stretches down to
$y_{\mathrm{bottom}}$. If $p_2$ and $p_3$ do not exist ($p_2$ and
$p_3$ are the same point if (4) only contains one maximal point), then
(2) and (5) are not computed, the leftmost slab of (3) is $k_1+1$, and
(4) stretches down to $p_4.y+1$.

To summarize, it follows that the skyline size for each multislab
query $R_\ell$ can be computed in $O(1)$ time, and the total time for a
skyline counting query becomes $O(\lg n/\lg\lg n)$.

\section{Skyline Range Reporting}
\label{sec:reporting}

In this section, we show how to extend our skyline range
counting data structure from Section~\ref{sec:upper-bound}
to also support reporting. Given a query
rectangle $R=[x_1,x_2] \times [y_1,y_2]$, we let $v_1,\dots,v_m$ and
$I_1,\dots,I_m$ be defined as in Section~\ref{sec:query}. The goal is to report
the skyline for each of the subqueries $R_\ell = I_\ell \times [z_\ell , y_2]$
where $z_1 = y_1$ and $z_\ell = \max\{z_{\ell-1},y_\ell^{\max}+1\}$ for
$\ell=2$ to $m$ and $y_\ell^{\max}$ is the maximal $y$-coordinate of a point in
$I_1 \times [1,y_2]$. Using the approach from Section~\ref{sec:query} we assume
the $z_\ell$'s have been computed as well as the index of the successor
of $y_1$ and the index of the predecessor of $z_\ell$ in each of the
$L_{v_\ell}$ lists. Recall the lists $L_{v}$ are not stored explicitly.

To answer the query $R_\ell$ at a node $v=v_\ell$, let $[i,j]$ be the range
of children of $v$ that are spanned by $R_\ell$ in the horizontal
direction and let $y_\bottom$ be the index of the successor of $z_\ell$
in $L_v$ and $y_\ttop$ the index of the predecessor of $y_2$ in
$L_v$. We first produce an output list $Y_\ell$ storing each point of
$\Skyline(R_\ell \cap P_v)$ as an index into $L_v$. The key observation
for producing this list is that the skyline inside a query rectangle
is the set of points produced by the following procedure: First report
the rightmost point in the query range and then recurse on the query
rectangle obtained by moving the bottom side of the query to just
above the returned point.

We implement this strategy in the following. First, if $R_\ell$ is
completely within one block (block $\lceil y_\ttop / \Delta^2
\rceil$ of $L_v$), we answer it by first running
$$\Rightmost(\sigma_v(\lceil y_\ttop/\Delta^2 \rceil),1+(y_\bottom-1
\bmod \Delta^2),1+(y_\ttop-1 \bmod \Delta^2),i,j)\;.$$ 
Adding
$\Delta^2\lceil y_\ttop / \Delta^2-1\rceil$ to the returned value
gives the index $k$ into $L_v$ of the rightmost point in the
output. We add $k$ to $Y_\ell$ and recurse on the query rectangle with
$y$-range from $k+1$ to $y_\ttop$.

If the query range is not contained in one block, we use the
decomposition into five queries that was introduced in
Section~\ref{sec:query}, see Figure~\ref{fig:multicolumn}. We define
$p_1,\dots,p_4$ and $(1),\dots,(5)$ as in Section~\ref{sec:query}. The
subquery $(1)$ is answered as just described for the case of a query
range completely within a block. The query $(4)$ is answered using
first the query
$\Rightmost_{c_v^{k_1}}(\Succ_v(p_2.y+1,k_1),\Pred_v(p_1.y,k_1))$. 
Following that, we move the bottom of the query rectangle just above
the returned point and recurse.

The query range~$(2)$ is answered by first repeatedly using the
$\Rightmost_{v,k_1+1,j}$ operation to identify the blocks within the
query range~$(2)$ containing points from $\Skyline(R_\ell\cap P_v)$.
Let $q_1,\ldots,q_t$ be the indexes into $L_v$ of the rightmost points
returned in each of these blocks, computable by
$$q_1=\Rightmost_{v,k_1+1,j}(\lceil y_\bottom/\Delta^2\rceil+1,\lceil
y_\ttop/\Delta^2 \rceil -1)\;,$$
and for $r>1$ (until no further point is found) by
$$q_r=\Rightmost_{v,k_1+1,j}(\lceil q_{r-1}/\Delta^2\rceil+1,\lceil
y_\ttop/\Delta^2 \rceil -1)\;.$$ Within each block $\lceil
q_r/\Delta^2\rceil$ we compute the additional points that should be
reported within slabs $[k_1,j]$ from right-to-left, starting
with $$\Rightmost(\sigma_v(\lceil q_r/\Delta^2\rceil),2+(q_r-1 \bmod
\Delta^2),\Delta^2,k_1,j)\;,$$ until no point is found or we find the
first point $f$ that should not be reported, i.e.\ $f$ is dominated by
$q_{r+1}$, which can be checked by the condition $\gamma<C_v(q_{r+1})$
or $\gamma=C_v(q_{r+1})$ and $q'=\Rightmost_{v'}(f',q')$, where
$\gamma=C_v(f)$, $v'=c_v^{\gamma}$ and $f'=\Pred_v(f,\gamma)$, and
$q'=\Pred_v(q_{r+1},\gamma)$.

The query $(5)$ is answered by repeatedly
using $\Rightmost_{c_v^{k_3}}$ and finally we answer $(3)$ using
$$\Rightmost(\sigma_v(\lceil y_\bottom/\Delta^2 \rceil),1+(y_\bottom-1
\bmod \Delta^2),\Delta^2,k_3+1,j)$$ and recursing above the returned
point. It follows that the list $Y_\ell$ is produced in
$O(1+|\Skyline(R_\ell \cap P_v)|) = O(1+|\Skyline(R_\ell \cap P)|)$
time. Summing over all lists $Y_\ell$, we get a total time of $O(\lg
n/\lg \lg n + k)$. 

What remains is to map the indices in the lists $Y_\ell$ to the actual
coordinates of the corresponding points. Using the $\pi$ arrays, this
can be done by repeatedly determining the position of $L_v[i]$ in
$L_{p_v}$. Doing this for all $O(\lg n/\lg \lg n)$ levels of the tree
allows one to deduce the global $y$-rank of the point corresponding to
$L_v[i]$. Storing an additional $O(n)$ sized array mapping global
$y$-ranks to the corresponding points gives a total running time of
$O((1+k) \lg n/\lg \lg n)$. To speed this up, we use the
Ball-Inheritance structure of~\cite{socg11}. For completeness, we describe
how this data structure is implemented in terms of the $\pi$ arrays we
have defined: Let $B \geq 2$ be a parameter. For every level $j$ in
the base tree $T$ that is a multiple of $B^i$, for $i=0,\dots,\lg_B
\lg_\Delta n$, we let all nodes $v$ at level $j$ store the following
array:

\begin{description}

\item [$\pi^{(B^i)}_v(j)$] For each $j, 1 \leq j \leq n_v$, stores the
  index of $L_v[j]$ in $L_{u(v)}$. Here $u(v)$ is the ancestor of $v$
  at the nearest level that is a multiple of $B^{i+1}$ (excluding
  possibly the level storing $v$). This can be supported by
  constructing the select data structure of Lemma~\ref{lem:rankselect} on the
  bit-vector $X$, where $X[j]=1$ if and only if $L_{u(v)}[j]$ is in
  $L_v$. A query $\pi^{(B^i)}_v(j)$ becomes $\select(j)$. The space usage
  for $\pi_v^{(B^i)}$ becomes
  $O(n_v \lg(n_{u(v)}/n_v)) = O(n_v \lg (\Delta^{B^{i+1}})) = O(n_v
  B^{i+1}\lg \Delta)$ bits.

\end{description}

Given an index $i$ into $L_v$, we can now recover $L_v[i]$ by using
the $\pi$ arrays to first jump $B$ levels up, then $B^2$ levels up and
so forth. The number of jumps becomes $O(\lg_B \lg_\Delta n)$ and hence
we get a query time of $O(\lg n/\lg \lg n + k\lg_B \lg_\Delta n)$. The
total space usage for all $\pi$ arrays becomes
$$
O\left(\sum_{i=1}^{\lg_B \lg_\Delta n} \frac{\lg_\Delta n}{B^i} \cdot B^{i+1} \lg \Delta \right) = 
O(n\lg n  \cdot (B \lg_B \lg_\Delta n))
$$ bits. Setting $B=\lg^{\varepsilon}n$ for an arbitrarily small
constant $\varepsilon>0$ gives a data structure with query time $O(\lg
n/\lg \lg n + k)$ and space usage $O(n \lg^{\varepsilon} n)$
words. Setting $B=2$ gives a data structure with query time $O(\lg
n/\lg \lg n+k\lg \lg n)$ and space usage $O(n \lg \lg n)$ words as
claimed.

\bibliographystyle{plain}
\bibliography{paper}

\end{document}

%% file: graph.pdf_t
\begin{picture}(0,0)%
\includegraphics{graph.pdf}%
\end{picture}%
\setlength{\unitlength}{3947sp}%
\begingroup\makeatletter\ifx\SetFigFont\undefined%
\gdef\SetFigFont#1#2#3#4#5{%
  \reset@font\fontsize{#1}{#2pt}%
  \fontfamily{#3}\fontseries{#4}\fontshape{#5}%
  \selectfont}%
\fi\endgroup%
\begin{picture}(4580,3220)(2854,-5882)
\put(3054,-3847){\makebox(0,0)[rb]{\smash{{\SetFigFont{8}{9.6}{\rmdefault}{\mddefault}{\updefault}{\color[rgb]{0,0,0}000}%
}}}}
\put(3654,-3847){\makebox(0,0)[rb]{\smash{{\SetFigFont{8}{9.6}{\rmdefault}{\mddefault}{\updefault}{\color[rgb]{0,0,0}001}%
}}}}
\put(4254,-3847){\makebox(0,0)[rb]{\smash{{\SetFigFont{8}{9.6}{\rmdefault}{\mddefault}{\updefault}{\color[rgb]{0,0,0}010}%
}}}}
\put(4854,-3847){\makebox(0,0)[rb]{\smash{{\SetFigFont{8}{9.6}{\rmdefault}{\mddefault}{\updefault}{\color[rgb]{0,0,0}011}%
}}}}
\put(5454,-3847){\makebox(0,0)[rb]{\smash{{\SetFigFont{8}{9.6}{\rmdefault}{\mddefault}{\updefault}{\color[rgb]{0,0,0}100}%
}}}}
\put(6054,-3847){\makebox(0,0)[rb]{\smash{{\SetFigFont{8}{9.6}{\rmdefault}{\mddefault}{\updefault}{\color[rgb]{0,0,0}101}%
}}}}
\put(6654,-3847){\makebox(0,0)[rb]{\smash{{\SetFigFont{8}{9.6}{\rmdefault}{\mddefault}{\updefault}{\color[rgb]{0,0,0}110}%
}}}}
\put(7254,-3847){\makebox(0,0)[rb]{\smash{{\SetFigFont{8}{9.6}{\rmdefault}{\mddefault}{\updefault}{\color[rgb]{0,0,0}111}%
}}}}
\put(3054,-4747){\makebox(0,0)[rb]{\smash{{\SetFigFont{8}{9.6}{\rmdefault}{\mddefault}{\updefault}{\color[rgb]{0,0,0}000}%
}}}}
\put(3654,-4747){\makebox(0,0)[rb]{\smash{{\SetFigFont{8}{9.6}{\rmdefault}{\mddefault}{\updefault}{\color[rgb]{0,0,0}001}%
}}}}
\put(4254,-4747){\makebox(0,0)[rb]{\smash{{\SetFigFont{8}{9.6}{\rmdefault}{\mddefault}{\updefault}{\color[rgb]{0,0,0}010}%
}}}}
\put(4854,-4747){\makebox(0,0)[rb]{\smash{{\SetFigFont{8}{9.6}{\rmdefault}{\mddefault}{\updefault}{\color[rgb]{0,0,0}011}%
}}}}
\put(5454,-4747){\makebox(0,0)[rb]{\smash{{\SetFigFont{8}{9.6}{\rmdefault}{\mddefault}{\updefault}{\color[rgb]{0,0,0}100}%
}}}}
\put(6054,-4747){\makebox(0,0)[rb]{\smash{{\SetFigFont{8}{9.6}{\rmdefault}{\mddefault}{\updefault}{\color[rgb]{0,0,0}101}%
}}}}
\put(6654,-4747){\makebox(0,0)[rb]{\smash{{\SetFigFont{8}{9.6}{\rmdefault}{\mddefault}{\updefault}{\color[rgb]{0,0,0}110}%
}}}}
\put(7254,-4747){\makebox(0,0)[rb]{\smash{{\SetFigFont{8}{9.6}{\rmdefault}{\mddefault}{\updefault}{\color[rgb]{0,0,0}111}%
}}}}
\put(3654,-5647){\makebox(0,0)[rb]{\smash{{\SetFigFont{8}{9.6}{\rmdefault}{\mddefault}{\updefault}{\color[rgb]{0,0,0}001}%
}}}}
\put(4254,-5647){\makebox(0,0)[rb]{\smash{{\SetFigFont{8}{9.6}{\rmdefault}{\mddefault}{\updefault}{\color[rgb]{0,0,0}010}%
}}}}
\put(4854,-5647){\makebox(0,0)[rb]{\smash{{\SetFigFont{8}{9.6}{\rmdefault}{\mddefault}{\updefault}{\color[rgb]{0,0,0}011}%
}}}}
\put(5454,-5647){\makebox(0,0)[rb]{\smash{{\SetFigFont{8}{9.6}{\rmdefault}{\mddefault}{\updefault}{\color[rgb]{0,0,0}100}%
}}}}
\put(6054,-5647){\makebox(0,0)[rb]{\smash{{\SetFigFont{8}{9.6}{\rmdefault}{\mddefault}{\updefault}{\color[rgb]{0,0,0}101}%
}}}}
\put(6654,-5647){\makebox(0,0)[rb]{\smash{{\SetFigFont{8}{9.6}{\rmdefault}{\mddefault}{\updefault}{\color[rgb]{0,0,0}110}%
}}}}
\put(7254,-5647){\makebox(0,0)[rb]{\smash{{\SetFigFont{8}{9.6}{\rmdefault}{\mddefault}{\updefault}{\color[rgb]{0,0,0}111}%
}}}}
\put(3654,-2947){\makebox(0,0)[rb]{\smash{{\SetFigFont{8}{9.6}{\rmdefault}{\mddefault}{\updefault}{\color[rgb]{0,0,0}001}%
}}}}
\put(4254,-2947){\makebox(0,0)[rb]{\smash{{\SetFigFont{8}{9.6}{\rmdefault}{\mddefault}{\updefault}{\color[rgb]{0,0,0}010}%
}}}}
\put(4854,-2947){\makebox(0,0)[rb]{\smash{{\SetFigFont{8}{9.6}{\rmdefault}{\mddefault}{\updefault}{\color[rgb]{0,0,0}011}%
}}}}
\put(5454,-2947){\makebox(0,0)[rb]{\smash{{\SetFigFont{8}{9.6}{\rmdefault}{\mddefault}{\updefault}{\color[rgb]{0,0,0}100}%
}}}}
\put(6054,-2947){\makebox(0,0)[rb]{\smash{{\SetFigFont{8}{9.6}{\rmdefault}{\mddefault}{\updefault}{\color[rgb]{0,0,0}101}%
}}}}
\put(6654,-2947){\makebox(0,0)[rb]{\smash{{\SetFigFont{8}{9.6}{\rmdefault}{\mddefault}{\updefault}{\color[rgb]{0,0,0}110}%
}}}}
\put(7254,-2947){\makebox(0,0)[rb]{\smash{{\SetFigFont{8}{9.6}{\rmdefault}{\mddefault}{\updefault}{\color[rgb]{0,0,0}111}%
}}}}
\put(3054,-2947){\makebox(0,0)[rb]{\smash{{\SetFigFont{8}{9.6}{\rmdefault}{\mddefault}{\updefault}{\color[rgb]{0,0,0}000}%
}}}}
\put(3054,-5647){\makebox(0,0)[rb]{\smash{{\SetFigFont{8}{9.6}{\rmdefault}{\mddefault}{\updefault}{\color[rgb]{0,0,0}000}%
}}}}
\put(3751,-2761){\makebox(0,0)[b]{\smash{{\SetFigFont{8}{9.6}{\rmdefault}{\mddefault}{\updefault}{\color[rgb]{0,0,0}$s$}%
}}}}
\put(6751,-5836){\makebox(0,0)[b]{\smash{{\SetFigFont{8}{9.6}{\rmdefault}{\mddefault}{\updefault}{\color[rgb]{0,0,0}$t$}%
}}}}
\put(3826,-3286){\makebox(0,0)[lb]{\smash{{\SetFigFont{8}{9.6}{\rmdefault}{\mddefault}{\updefault}{\color[rgb]{0,0,0}$a$}%
}}}}
\put(3451,-4036){\makebox(0,0)[lb]{\smash{{\SetFigFont{8}{9.6}{\rmdefault}{\mddefault}{\updefault}{\color[rgb]{0,0,0}$b$}%
}}}}
\put(6976,-5011){\makebox(0,0)[lb]{\smash{{\SetFigFont{8}{9.6}{\rmdefault}{\mddefault}{\updefault}{\color[rgb]{0,0,0}$c$}%
}}}}
\end{picture}%

%% file: rectangles2.pdf_t
\begin{picture}(0,0)%
\includegraphics{rectangles2.pdf}%
\end{picture}%
\setlength{\unitlength}{868sp}%
\begingroup\makeatletter\ifx\SetFigFont\undefined%
\gdef\SetFigFont#1#2#3#4#5{%
  \reset@font\fontsize{#1}{#2pt}%
  \fontfamily{#3}\fontseries{#4}\fontshape{#5}%
  \selectfont}%
\fi\endgroup%
\begin{picture}(21123,21175)(3211,-20903)
\put(3226,-10561){\makebox(0,0)[rb]{\smash{{\SetFigFont{11}{13.2}{\rmdefault}{\mddefault}{\updefault}{\color[rgb]{0,0,0}rev$_2(t)$}%
}}}}
\put(8776,-20686){\makebox(0,0)[b]{\smash{{\SetFigFont{11}{13.2}{\rmdefault}{\mddefault}{\updefault}{\color[rgb]{0,0,0}$s$}%
}}}}
\put(6526,-19561){\makebox(0,0)[b]{\smash{{\SetFigFont{11}{13.2}{\rmdefault}{\mddefault}{\updefault}{\color[rgb]{0,0,0}000}%
}}}}
\put(8776,-19561){\makebox(0,0)[b]{\smash{{\SetFigFont{11}{13.2}{\rmdefault}{\mddefault}{\updefault}{\color[rgb]{0,0,0}001}%
}}}}
\put(11326,-19561){\makebox(0,0)[b]{\smash{{\SetFigFont{11}{13.2}{\rmdefault}{\mddefault}{\updefault}{\color[rgb]{0,0,0}010}%
}}}}
\put(13576,-19561){\makebox(0,0)[b]{\smash{{\SetFigFont{11}{13.2}{\rmdefault}{\mddefault}{\updefault}{\color[rgb]{0,0,0}011}%
}}}}
\put(16201,-19561){\makebox(0,0)[b]{\smash{{\SetFigFont{11}{13.2}{\rmdefault}{\mddefault}{\updefault}{\color[rgb]{0,0,0}100}%
}}}}
\put(18376,-19561){\makebox(0,0)[b]{\smash{{\SetFigFont{11}{13.2}{\rmdefault}{\mddefault}{\updefault}{\color[rgb]{0,0,0}101}%
}}}}
\put(20926,-19561){\makebox(0,0)[b]{\smash{{\SetFigFont{11}{13.2}{\rmdefault}{\mddefault}{\updefault}{\color[rgb]{0,0,0}110}%
}}}}
\put(23176,-19561){\makebox(0,0)[b]{\smash{{\SetFigFont{11}{13.2}{\rmdefault}{\mddefault}{\updefault}{\color[rgb]{0,0,0}111}%
}}}}
\put(5101,-961){\makebox(0,0)[rb]{\smash{{\SetFigFont{11}{13.2}{\rmdefault}{\mddefault}{\updefault}{\color[rgb]{0,0,0}111}%
}}}}
\put(5101,-17761){\makebox(0,0)[rb]{\smash{{\SetFigFont{11}{13.2}{\rmdefault}{\mddefault}{\updefault}{\color[rgb]{0,0,0}000}%
}}}}
\put(5101,-15361){\makebox(0,0)[rb]{\smash{{\SetFigFont{11}{13.2}{\rmdefault}{\mddefault}{\updefault}{\color[rgb]{0,0,0}001}%
}}}}
\put(5101,-12961){\makebox(0,0)[rb]{\smash{{\SetFigFont{11}{13.2}{\rmdefault}{\mddefault}{\updefault}{\color[rgb]{0,0,0}010}%
}}}}
\put(5101,-10561){\makebox(0,0)[rb]{\smash{{\SetFigFont{11}{13.2}{\rmdefault}{\mddefault}{\updefault}{\color[rgb]{0,0,0}011}%
}}}}
\put(5101,-8161){\makebox(0,0)[rb]{\smash{{\SetFigFont{11}{13.2}{\rmdefault}{\mddefault}{\updefault}{\color[rgb]{0,0,0}100}%
}}}}
\put(5101,-5761){\makebox(0,0)[rb]{\smash{{\SetFigFont{11}{13.2}{\rmdefault}{\mddefault}{\updefault}{\color[rgb]{0,0,0}101}%
}}}}
\put(5101,-3361){\makebox(0,0)[rb]{\smash{{\SetFigFont{11}{13.2}{\rmdefault}{\mddefault}{\updefault}{\color[rgb]{0,0,0}110}%
}}}}
\put(9376,-586){\makebox(0,0)[b]{\smash{{\SetFigFont{11}{13.2}{\rmdefault}{\mddefault}{\updefault}{\color[rgb]{0,0,0}$a$}%
}}}}
\put(16726,-2836){\makebox(0,0)[b]{\smash{{\SetFigFont{11}{13.2}{\rmdefault}{\mddefault}{\updefault}{\color[rgb]{0,0,0}$c$}%
}}}}
\put(7126,-10186){\makebox(0,0)[b]{\smash{{\SetFigFont{11}{13.2}{\rmdefault}{\mddefault}{\updefault}{\color[rgb]{0,0,0}$b$}%
}}}}
\end{picture}%

%% file: rectangles3.pdf_t
\begin{picture}(0,0)%
\includegraphics{rectangles3.pdf}%
\end{picture}%
\setlength{\unitlength}{868sp}%
\begingroup\makeatletter\ifx\SetFigFont\undefined%
\gdef\SetFigFont#1#2#3#4#5{%
  \reset@font\fontsize{#1}{#2pt}%
  \fontfamily{#3}\fontseries{#4}\fontshape{#5}%
  \selectfont}%
\fi\endgroup%
\begin{picture}(21123,21175)(3211,-20903)
\put(3226,-10561){\makebox(0,0)[rb]{\smash{{\SetFigFont{11}{13.2}{\rmdefault}{\mddefault}{\updefault}{\color[rgb]{0,0,0}rev$_2(t)$}%
}}}}
\put(8776,-20686){\makebox(0,0)[b]{\smash{{\SetFigFont{11}{13.2}{\rmdefault}{\mddefault}{\updefault}{\color[rgb]{0,0,0}$s$}%
}}}}
\put(6526,-19561){\makebox(0,0)[b]{\smash{{\SetFigFont{11}{13.2}{\rmdefault}{\mddefault}{\updefault}{\color[rgb]{0,0,0}000}%
}}}}
\put(8776,-19561){\makebox(0,0)[b]{\smash{{\SetFigFont{11}{13.2}{\rmdefault}{\mddefault}{\updefault}{\color[rgb]{0,0,0}001}%
}}}}
\put(11326,-19561){\makebox(0,0)[b]{\smash{{\SetFigFont{11}{13.2}{\rmdefault}{\mddefault}{\updefault}{\color[rgb]{0,0,0}010}%
}}}}
\put(13576,-19561){\makebox(0,0)[b]{\smash{{\SetFigFont{11}{13.2}{\rmdefault}{\mddefault}{\updefault}{\color[rgb]{0,0,0}011}%
}}}}
\put(16201,-19561){\makebox(0,0)[b]{\smash{{\SetFigFont{11}{13.2}{\rmdefault}{\mddefault}{\updefault}{\color[rgb]{0,0,0}100}%
}}}}
\put(18376,-19561){\makebox(0,0)[b]{\smash{{\SetFigFont{11}{13.2}{\rmdefault}{\mddefault}{\updefault}{\color[rgb]{0,0,0}101}%
}}}}
\put(20926,-19561){\makebox(0,0)[b]{\smash{{\SetFigFont{11}{13.2}{\rmdefault}{\mddefault}{\updefault}{\color[rgb]{0,0,0}110}%
}}}}
\put(23176,-19561){\makebox(0,0)[b]{\smash{{\SetFigFont{11}{13.2}{\rmdefault}{\mddefault}{\updefault}{\color[rgb]{0,0,0}111}%
}}}}
\put(5101,-961){\makebox(0,0)[rb]{\smash{{\SetFigFont{11}{13.2}{\rmdefault}{\mddefault}{\updefault}{\color[rgb]{0,0,0}111}%
}}}}
\put(5101,-17761){\makebox(0,0)[rb]{\smash{{\SetFigFont{11}{13.2}{\rmdefault}{\mddefault}{\updefault}{\color[rgb]{0,0,0}000}%
}}}}
\put(5101,-15361){\makebox(0,0)[rb]{\smash{{\SetFigFont{11}{13.2}{\rmdefault}{\mddefault}{\updefault}{\color[rgb]{0,0,0}001}%
}}}}
\put(5101,-12961){\makebox(0,0)[rb]{\smash{{\SetFigFont{11}{13.2}{\rmdefault}{\mddefault}{\updefault}{\color[rgb]{0,0,0}010}%
}}}}
\put(5101,-10561){\makebox(0,0)[rb]{\smash{{\SetFigFont{11}{13.2}{\rmdefault}{\mddefault}{\updefault}{\color[rgb]{0,0,0}011}%
}}}}
\put(5101,-8161){\makebox(0,0)[rb]{\smash{{\SetFigFont{11}{13.2}{\rmdefault}{\mddefault}{\updefault}{\color[rgb]{0,0,0}100}%
}}}}
\put(5101,-5761){\makebox(0,0)[rb]{\smash{{\SetFigFont{11}{13.2}{\rmdefault}{\mddefault}{\updefault}{\color[rgb]{0,0,0}101}%
}}}}
\put(5101,-3361){\makebox(0,0)[rb]{\smash{{\SetFigFont{11}{13.2}{\rmdefault}{\mddefault}{\updefault}{\color[rgb]{0,0,0}110}%
}}}}
\put(7126,-10186){\makebox(0,0)[b]{\smash{{\SetFigFont{11}{13.2}{\rmdefault}{\mddefault}{\updefault}{\color[rgb]{0,0,0}$b$}%
}}}}
\put(9376,-586){\makebox(0,0)[b]{\smash{{\SetFigFont{11}{13.2}{\rmdefault}{\mddefault}{\updefault}{\color[rgb]{0,0,0}$a$}%
}}}}
\put(16726,-2911){\makebox(0,0)[b]{\smash{{\SetFigFont{11}{13.2}{\rmdefault}{\mddefault}{\updefault}{\color[rgb]{0,0,0}$c$}%
}}}}
\end{picture}%

%% file: decomposition.pdf_t
\begin{picture}(0,0)%
\includegraphics{decomposition.pdf}%
\end{picture}%
\setlength{\unitlength}{3108sp}%
\begingroup\makeatletter\ifx\SetFigFont\undefined%
\gdef\SetFigFont#1#2#3#4#5{%
  \reset@font\fontsize{#1}{#2pt}%
  \fontfamily{#3}\fontseries{#4}\fontshape{#5}%
  \selectfont}%
\fi\endgroup%
\begin{picture}(7002,5031)(4984,-7739)
\put(11971,-7036){\makebox(0,0)[b]{\smash{{\SetFigFont{9}{10.8}{\rmdefault}{\mddefault}{\updefault}{\color[rgb]{0,0,0}$y_1$}%
}}}}
\put(11971,-4876){\makebox(0,0)[b]{\smash{{\SetFigFont{9}{10.8}{\rmdefault}{\mddefault}{\updefault}{\color[rgb]{0,0,0}$y_2$}%
}}}}
\put(11026,-5056){\makebox(0,0)[b]{\smash{{\SetFigFont{9}{10.8}{\rmdefault}{\mddefault}{\updefault}{\color[rgb]{0,0,0}$R_1$}%
}}}}
\put(10531,-5056){\makebox(0,0)[b]{\smash{{\SetFigFont{9}{10.8}{\rmdefault}{\mddefault}{\updefault}{\color[rgb]{0,0,0}$R_2$}%
}}}}
\put(8506,-5056){\makebox(0,0)[b]{\smash{{\SetFigFont{9}{10.8}{\rmdefault}{\mddefault}{\updefault}{\color[rgb]{0,0,0}$R_3$}%
}}}}
\put(6301,-5056){\makebox(0,0)[b]{\smash{{\SetFigFont{9}{10.8}{\rmdefault}{\mddefault}{\updefault}{\color[rgb]{0,0,0}$R_4$}%
}}}}
\put(5356,-3076){\makebox(0,0)[b]{\smash{{\SetFigFont{9}{10.8}{\rmdefault}{\mddefault}{\updefault}{\color[rgb]{0,0,0}$T$}%
}}}}
\put(5401,-4831){\makebox(0,0)[b]{\smash{{\SetFigFont{9}{10.8}{\rmdefault}{\mddefault}{\updefault}{\color[rgb]{0,0,0}$R$}%
}}}}
\put(5716,-4966){\makebox(0,0)[b]{\smash{{\SetFigFont{9}{10.8}{\rmdefault}{\mddefault}{\updefault}{\color[rgb]{0,0,0}$R_5$}%
}}}}
\put(8461,-7666){\makebox(0,0)[b]{\smash{{\SetFigFont{9}{10.8}{\rmdefault}{\mddefault}{\updefault}{\color[rgb]{0,0,0}$I_3$}%
}}}}
\put(5761,-7666){\makebox(0,0)[b]{\smash{{\SetFigFont{9}{10.8}{\rmdefault}{\mddefault}{\updefault}{\color[rgb]{0,0,0}$I_5$}%
}}}}
\put(11071,-7666){\makebox(0,0)[b]{\smash{{\SetFigFont{9}{10.8}{\rmdefault}{\mddefault}{\updefault}{\color[rgb]{0,0,0}$I_1$}%
}}}}
\put(10576,-7666){\makebox(0,0)[b]{\smash{{\SetFigFont{9}{10.8}{\rmdefault}{\mddefault}{\updefault}{\color[rgb]{0,0,0}$I_2$}%
}}}}
\put(6346,-7666){\makebox(0,0)[b]{\smash{{\SetFigFont{9}{10.8}{\rmdefault}{\mddefault}{\updefault}{\color[rgb]{0,0,0}$I_4$}%
}}}}
\put(11161,-7396){\makebox(0,0)[b]{\smash{{\SetFigFont{9}{10.8}{\rmdefault}{\mddefault}{\updefault}{\color[rgb]{0,0,0}$x_2$}%
}}}}
\put(5581,-7396){\makebox(0,0)[b]{\smash{{\SetFigFont{9}{10.8}{\rmdefault}{\mddefault}{\updefault}{\color[rgb]{0,0,0}$x_1$}%
}}}}
\put(11071,-3706){\makebox(0,0)[b]{\smash{{\SetFigFont{9}{10.8}{\rmdefault}{\mddefault}{\updefault}{\color[rgb]{0,0,0}$v_1$}%
}}}}
\put(10801,-3256){\makebox(0,0)[b]{\smash{{\SetFigFont{9}{10.8}{\rmdefault}{\mddefault}{\updefault}{\color[rgb]{0,0,0}$v_2$}%
}}}}
\put(5581,-3706){\makebox(0,0)[b]{\smash{{\SetFigFont{9}{10.8}{\rmdefault}{\mddefault}{\updefault}{\color[rgb]{0,0,0}$v_5$}%
}}}}
\put(8551,-2896){\makebox(0,0)[b]{\smash{{\SetFigFont{9}{10.8}{\rmdefault}{\mddefault}{\updefault}{\color[rgb]{0,0,0}$v_3$}%
}}}}
\put(5941,-3256){\makebox(0,0)[b]{\smash{{\SetFigFont{9}{10.8}{\rmdefault}{\mddefault}{\updefault}{\color[rgb]{0,0,0}$v_4$}%
}}}}
\end{picture}%

%% file: skyline-sum.pdf_t
\begin{picture}(0,0)%
\includegraphics{skyline-sum.pdf}%
\end{picture}%
\setlength{\unitlength}{3108sp}%
\begingroup\makeatletter\ifx\SetFigFont\undefined%
\gdef\SetFigFont#1#2#3#4#5{%
  \reset@font\fontsize{#1}{#2pt}%
  \fontfamily{#3}\fontseries{#4}\fontshape{#5}%
  \selectfont}%
\fi\endgroup%
\begin{picture}(2144,1785)(5704,-5431)
\end{picture}%

%% file: skyline-diff.pdf_t
\begin{picture}(0,0)%
\includegraphics{skyline-diff.pdf}%
\end{picture}%
\setlength{\unitlength}{3108sp}%
\begingroup\makeatletter\ifx\SetFigFont\undefined%
\gdef\SetFigFont#1#2#3#4#5{%
  \reset@font\fontsize{#1}{#2pt}%
  \fontfamily{#3}\fontseries{#4}\fontshape{#5}%
  \selectfont}%
\fi\endgroup%
\begin{picture}(2055,1733)(5563,-5431)
\put(5671,-3976){\makebox(0,0)[rb]{\smash{{\SetFigFont{9}{10.8}{\rmdefault}{\mddefault}{\updefault}{\color[rgb]{0,0,0}$j$}%
}}}}
\put(5671,-4516){\makebox(0,0)[rb]{\smash{{\SetFigFont{9}{10.8}{\rmdefault}{\mddefault}{\updefault}{\color[rgb]{0,0,0}$k$}%
}}}}
\put(5671,-4876){\makebox(0,0)[rb]{\smash{{\SetFigFont{9}{10.8}{\rmdefault}{\mddefault}{\updefault}{\color[rgb]{0,0,0}$i$}%
}}}}
\end{picture}%

%% file: multicolumn.pdf_t
\begin{picture}(0,0)%
\includegraphics{multicolumn.pdf}%
\end{picture}%
\setlength{\unitlength}{2693sp}%
\begingroup\makeatletter\ifx\SetFigFont\undefined%
\gdef\SetFigFont#1#2#3#4#5{%
  \reset@font\fontsize{#1}{#2pt}%
  \fontfamily{#3}\fontseries{#4}\fontshape{#5}%
  \selectfont}%
\fi\endgroup%
\begin{picture}(9750,5015)(4126,-12583)
\put(4141,-9016){\makebox(0,0)[rb]{\smash{{\SetFigFont{8}{9.6}{\rmdefault}{\mddefault}{\updefault}{\color[rgb]{0,0,0}$B_\mathrm{top}$}%
}}}}
\put(4141,-12166){\makebox(0,0)[rb]{\smash{{\SetFigFont{8}{9.6}{\rmdefault}{\mddefault}{\updefault}{\color[rgb]{0,0,0}$B_\mathrm{bottom}$}%
}}}}
\put(13861,-12166){\makebox(0,0)[lb]{\smash{{\SetFigFont{8}{9.6}{\rmdefault}{\mddefault}{\updefault}{\color[rgb]{0,0,0}$y_\mathrm{bottom}$}%
}}}}
\put(13861,-8836){\makebox(0,0)[lb]{\smash{{\SetFigFont{8}{9.6}{\rmdefault}{\mddefault}{\updefault}{\color[rgb]{0,0,0}$y_\mathrm{top}$}%
}}}}
\put(13726,-10276){\makebox(0,0)[lb]{\smash{{\SetFigFont{8}{9.6}{\rmdefault}{\mddefault}{\updefault}{\color[rgb]{0,0,0}$\Delta^2$}%
}}}}
\put(8911,-7756){\makebox(0,0)[b]{\smash{{\SetFigFont{8}{9.6}{\rmdefault}{\mddefault}{\updefault}{\color[rgb]{0,0,0}$v$}%
}}}}
\put(5986,-8476){\makebox(0,0)[b]{\smash{{\SetFigFont{8}{9.6}{\rmdefault}{\mddefault}{\updefault}{\color[rgb]{0,0,0}$c_v^i$}%
}}}}
\put(11836,-8476){\makebox(0,0)[b]{\smash{{\SetFigFont{8}{9.6}{\rmdefault}{\mddefault}{\updefault}{\color[rgb]{0,0,0}$c_v^j$}%
}}}}
\put(7156,-8476){\makebox(0,0)[b]{\smash{{\SetFigFont{8}{9.6}{\rmdefault}{\mddefault}{\updefault}{\color[rgb]{0,0,0}$c_v^{k_1}$}%
}}}}
\put(9496,-8476){\makebox(0,0)[b]{\smash{{\SetFigFont{8}{9.6}{\rmdefault}{\mddefault}{\updefault}{\color[rgb]{0,0,0}$c_v^{k_3}$}%
}}}}
\put(9451,-10096){\makebox(0,0)[lb]{\smash{{\SetFigFont{8}{9.6}{\rmdefault}{\mddefault}{\updefault}{\color[rgb]{0,0,0}$p_3$}%
}}}}
\put(11566,-11896){\makebox(0,0)[lb]{\smash{{\SetFigFont{8}{9.6}{\rmdefault}{\mddefault}{\updefault}{\color[rgb]{0,0,0}$p_4$}%
}}}}
\put(8011,-9511){\makebox(0,0)[lb]{\smash{{\SetFigFont{8}{9.6}{\rmdefault}{\mddefault}{\updefault}{\color[rgb]{0,0,0}$p_2$}%
}}}}
\put(7246,-8926){\makebox(0,0)[lb]{\smash{{\SetFigFont{8}{9.6}{\rmdefault}{\mddefault}{\updefault}{\color[rgb]{0,0,0}$p_1$}%
}}}}
\put(9496,-10861){\makebox(0,0)[b]{\smash{{\SetFigFont{8}{9.6}{\rmdefault}{\mddefault}{\updefault}{\color[rgb]{0,0,0}(5)}%
}}}}
\put(10666,-10231){\makebox(0,0)[b]{\smash{{\SetFigFont{8}{9.6}{\rmdefault}{\mddefault}{\updefault}{\color[rgb]{0,0,0}(2)}%
}}}}
\put(9496,-9061){\makebox(0,0)[b]{\smash{{\SetFigFont{8}{9.6}{\rmdefault}{\mddefault}{\updefault}{\color[rgb]{0,0,0}(1)}%
}}}}
\put(10666,-11986){\makebox(0,0)[b]{\smash{{\SetFigFont{8}{9.6}{\rmdefault}{\mddefault}{\updefault}{\color[rgb]{0,0,0}(3)}%
}}}}
\put(7111,-9331){\makebox(0,0)[b]{\smash{{\SetFigFont{8}{9.6}{\rmdefault}{\mddefault}{\updefault}{\color[rgb]{0,0,0}(4)}%
}}}}
\end{picture}%